\pdfoutput=1
\documentclass[a4paper,12pt]{article}
\usepackage[utf8]{inputenc}
\usepackage{amsmath,amssymb,amsthm}
\usepackage{mathtools}
\usepackage{dsfont}
\usepackage[T1]{fontenc}
\usepackage[ngerman,english]{babel}
\usepackage[inline]{enumitem}
\usepackage{graphicx}
\usepackage{cite}
\usepackage{ulem}
\usepackage{etoolbox}
\apptocmd{\sloppy}{\hbadness 10000\relax}{}{}
\usepackage{xcolor}
\usepackage{dashbox}
\usepackage{asymptotic}

\numberwithin{equation}{section}

\newcommand{\EE}{\mathcal{E}}

\newcommand{\HH}{\mathcal{H}}

\newcommand{\FF}{\mathcal{F}}
\newcommand{\CC}{\mathcal{C}}
\newcommand{\VV}{\mathcal{V}}

\newcommand{\hh}{\mathfrak{h}}

\newcommand{\RR}{\mathcal{R}}%
\newcommand{\R}{{\mathord{\mathbb R}}}%
\newcommand{\Z}{{\mathord{\mathbb Z}}}%
\newcommand{\N}{{\mathord{\mathbb N}}}%
\newcommand{\C}{{\mathord{\mathbb C}}}%
\newtheorem{theorem}{Theorem}
\newtheorem{proposition}{Proposition}
\newtheorem{lemma}{Lemma}
\newtheorem{corollary}{Corollary}
\theoremstyle{remark}
\newtheorem{remark}{Remark}

\newtheorem{example}{Example}
\theoremstyle{definition}
\newtheorem{definition}{Definition}
\newtheorem{hypothesis}{Hypothesis}

\DeclareMathOperator*{\subst}{subst}

\newcommand{\ran}{{\rm Ran}}
\newcommand{\dom}{{\rm dom}}

\setlist[enumerate,1]{label=(\arabic*)}
\newcommand{\proofparagraph}[1]{\par\medskip\noindent{\bf\footnotesize
    #1}}

\newcommand{\inn}[1]{\langle {#1} \rangle }

\newcommand{\bignode}[2]{\tikz[baseline=(char.base)]{
    \node[fill=white,circle,draw,inner
    sep=0.8pt](char){\begin{minipage}[c][5mm]{5mm}\centering
        $#1$ \end{minipage}}; \node[above] at (char.north) {$#2$};}\ }
\newcommand{\node}[2]{\tikz[baseline=(char.base)]{
    \node[fill=white,circle,draw,inner sep=0.8pt](char){$#1$}; \node[above] at (char.north) {$#2$};}}

\newcommand{\edge}[1]{\mspace{-8mu}\frac{\quad\raisebox{.8ex}{$#1$}\quad}{}\mspace{-8mu}}

\newcommand*{\colorboxed}{}
\def\colorboxed#1#{%
  \colorboxedAux{#1}%
}
\newcommand*{\colorboxedAux}[3]{%
  \begingroup
    \colorlet{cb@saved}{.}%
    \color#1{#2}%
    \boxed{%
      \color{cb@saved}%
      #3%
    }%
  \endgroup
}

\begin{document}

\title{On Asymptotic Expansions in Spin Boson  \mbox{Models}}
\author{Gerhard Br\"aunlich, David Hasler, Markus Lange
  \\ \\  Mathematical Institute, University of Jena \\
  Ernst-Abbe-Platz 2, 07743 Jena, Germany}

\date{\today}

\maketitle

\begin{abstract}
We consider expansions of eigenvalues and eigenvectors of models of
quantum field theory. For a class of models known as generalized spin
boson model we prove the existence of asymptotic expansions of the
ground state and the ground state energy to arbitrary order. We need a
mild but very natural infrared assumption, which is weaker than the assumption
usually needed for other methods such as operator theoretic renormalization to
be applicable. The result complements
previously shown analyticity properties.
\end{abstract}


\section{Introduction}

Perturbation theory is widely used to calculate various
quantities in quantum mechanics.
As long as the perturbation is ``small'' compared to the unperturbed
system one expects to obtain  good  approximations to physical quantities.
In particular, in case of isolated eigenvalues analytic perturbation theory is
available  allowing the calculation of eigenvalues and eigenvectors
in terms of convergent power series, which are also known as
Rayleigh-Schr\"odinger perturbation series \cite{ReeSim4,Kat95}. However,
in many-body quantum systems and models of massless quantum fields the
ground state is typically not isolated from the rest of the spectrum and analytic
perturbation theory is not applicable.  Different methods to cope with these problems have been developed,
see for example \cite{GriHas09,AbdHas12,CatHai04,BacFroPiz06,BacFroPiz09} or references mentioned below.

In this paper  we consider models of massless quantum fields.
Specifically, we consider a  quantum mechanical
system with finitely many degrees of freedom, which is linearly coupled
to a field of relativistic massless bosons.
Such models are also known
as generalized spin-boson models.  They are  used to describe low energy aspects
of non-relativistic quantum mechanical matter interacting with a
quantized radiation field such as a field of phonons or a field of photons.
Various spectral properties of the Hamiltonians of such models have
been investigated.
In particular, we assume that the quantum field is  massless.
This implies that the  ground state energy as well as resonance energies are
 not isolated from the rest of the spectrum.
Existence of ground states and resonance states have been shown to exist for such models
\cite{Fro74,BacFroSig98-1,BacFroSig98-2,BacFroSig99,Ger00,GriLieLos01,AbdHas12}.
In spite that  the ground state energy is embedded in continuous spectrum and analytic perturbation theory is
not applicable, it has been shown  in various situations that  the ground state and the ground state energy are in fact  analytic
functions of the coupling constant
\cite{GriHas09,HasHer11-1,Abd11,AbdHas12}. To prove  these results
one uses operator theoretic renormalization
\cite{BCFS} and in some cases on can employ expansion techniques from  statistical mechanics.
The analyticity  results obtained by renormalization  are  rather  surprising.
The calculation of the Rayleigh-Schr\"odinger expansion coefficients involve sums of divergent expressions,
and it is at first sight not obvious in which  situations these infinities will eventually cancel each other.
On the other hand there exist situations where the ground state energy is not an analytic
function of the coupling constant \cite{BarCheVouVug10}.

In this paper we show that for a large class of generalized spin boson models
there exist   asymptotic expansions for the ground state and the ground state energy
to arbitrary order. Whereas the existence of  asymptotic expansion is
weaker than the existence of an analytic expansion, our result holds
in situations where analytic expansions have  not been shown.
We expect that our technique can be used to derive asymptotic
expansions in situations where analytic expansions in fact do not exist.
Such a situation may  occur when  the unperturbed
operator has a degenerate ground state energy, which is lifted once the
interaction is turned on. This will be addressed in a forthcoming paper by the authors.

We want to mention that for models which we consider asymptotic expansions
have been investigated in several papers. In particular expansions
of  the first few orders have been investigated in
 \cite{HaiSei02,BarCheVug03,BarCheVouVug08,BarCheVouVug10}.
More recently  in \cite{Ara14}  asymptotic expansion formulas
have been  studied  to   arbitrary order, provided the   infrared regularization
is sufficiently strong, i.e., the higher the order of  expansion
the stronger the infrared regularization.
In the present paper we relax this infrared assumption substantially.
Our main result of the paper,  Theorem~\ref{thm:groundstatee}, stated below,
shows the existence of an asymptotic expansion for a reasonable
infrared assumption. The key idea in the proof is to show
that the infinities involved in calculating the Rayleigh-Schr\"odinger
expansion coefficients cancel out.
Showing that these cancellations can be controlled to arbitrary order,
without any analyticity assumption, is the main new technical contribution of the present  paper.

The paper is organized as follows.
In the next section we introduce the model and state the main result.
In Section~\ref{sec:pertheory} we  derive for a general
class of models
formulas for   expansion coefficients  of the
ground state and the ground state energy in terms
of the coupling constant. Assuming that the
expansion coefficients are finite, which will be shown
in Sections~\ref{sec:groundstatee} and  \ref{sec:groundstate},
we determine general conditions for which these expansions coefficients yield
an asymptotic expansion.

In Section~\ref{sec:groundstatee}  we show  Theorem~\ref{thm:mainenergy}, i.e.,
the finiteness of the expansion coefficients of the ground
state energy.
To this end,  we first express the expansion coefficients as a sum  of
linked contractions involving renormalized propagators, which we call
renormalized Feynman  graphs.  The renormalized
propagators take into account the cancellations which
results in an improved infrared behaviour.
Finally we estimate the  renormalized
Feynman graphs  and prove  the finiteness of
each expansion coefficient.

Assuming a certain condition we show in
Section~\ref{sec:groundstate} the finiteness of the expansion coefficients for
the ground state. Similarly to Section~\ref{sec:groundstatee}
we first express the squared of the norm of the
expansion coefficients as a sum  of
linked contractions involving renormalized propagators, except the
one in the middle. We then use that formula to show the finiteness of the
expansion coefficients of the ground state.
\enlargethispage{1cm}

In Section~\ref{sec:proofofmain} we collect the results of the previous
sections and provide a proof of Theorem~\ref{thm:groundstatee}.

\section{Model and Statement of Main Results}
In this section we introduce the model and state the main result.
 Let $\HH_{\rm at}$ be a separable Hilbert
space and  let $H_{\rm at}$ be a selfadjoint operator in $\HH_{\rm at}$.
Assume that  $E_{\rm at} = \inf \sigma (H_{\rm at})$
is a nondegenerate  eigenvalue of $H_{\rm at}$, which is isolated
from the rest of the spectrum,  i.e.,
$$
E_{\rm at} < \epsilon_1 := \inf ( \sigma(H_{\rm at}) \setminus \{ E_{\rm at} \} ) .
$$
Let  $\varphi_{\rm at}$ denote the  normalized eigenvector
and let $P_{\rm at}$ denote the orthogonal  eigenprojection of $E_{\rm at}$.
For a separable Hilbert space $\hh$ we write
\begin{align*}
  L^2_s ((\R^3)^{n} ; \hh )  :=   \big\{ \psi \in L^2((\R^3)^{n} ; \hh ) : \psi(k_1,&\dotsc,k_n) =    \psi(k_{\pi(1)},\dotsc,k_{\pi(n)} )\\
  &\forall  \text{ permutations } \pi \text { of } \{1,\ldots,n
  \} \big\}.
\end{align*}
We introduce the symmetric Fock space
$$
\FF = \bigoplus_{n=0}^\infty \FF_n ,
$$
where the so called $n$-photon subspaces are  defined by
\begin{align*}
  \FF_0 & := \C   , \\
  \FF_n  & :=  L^2_s ((\R^3)^{n }; \C ).
\end{align*}
We introduce  the so called vacuum vector $ \Omega = (1,0,0,\cdots )  \in \FF$.
The free field Hamiltonian is defined by
\begin{align*}
  H_f : \dom(H_f) \subset \FF  &\to   \FF  \\
  (H_f \psi)_n(k_1,\dotsc,k_n)   &:=   (|k_1|+|k_2|+ \cdots + |k_n|)  \psi_n(k_1,\dotsc,k_n) ,
\end{align*}
where $\dom(H_f) := \{ \psi \in \FF : H_f \psi \in \FF \} $.  The total
Hilbert space is defined by
$$
\HH := \HH_{\rm at} \otimes \FF \simeq \bigoplus_{n=0}^\infty
L_s^2((\R^{3})^n ; \HH_{\rm at} ) .
$$
We shall identify the spaces on the right hand side and occasionaly drop the tensor sign in the notation.
For $G : \R^3 \to \mathcal{L}(\HH_{\rm at}) $ a strongly  measurable
function such that
$$
\int \| G(k) \|^2  dk < \infty ,
$$
we define the so called annihilation operator
\begin{align*}
 a(G) &: \dom(a(G)) \subset \HH  \to \HH  \\
  & \psi \mapsto  (a(G) \psi )_n(k_1,\dotsc,k_n): = \sqrt{n+1} \int G^*(k)\psi_{n+1}(k, k_1,\dotsc,k_n) dk     , 
\end{align*}
where $\dom(a(G)) := \{ \psi \in \HH : a(G) \psi \in \HH \}$. One readily verifies that  $a(G)$ is a densely defined
closed operator. We denote its adjoint by
$
a^*(G) := a(G)^* ,
$
and introduce the field operator by
$$
 \phi(G) :=  \overline{ a(G) + a^*(G) }   ,
$$
where the line denotes the closure. \\
To define  the total Hamiltonian  we assume in addition that
\begin{equation} \label{eq:assonG}
\int \| G(k) \|^2(1+|k|^{-1}) dk < \infty ,
\end{equation}
since then it is well known that $\phi(G)$ is
infinitesimally   small  with respect to
$1_{\HH_{\rm at}} \otimes H_f $. This allows us to
 define the total  Hamiltonian of the interacting system by
\begin{equation} \label{eq:defofh}
H(\lambda) = H_{\rm at} \otimes 1_\FF + 1_{\HH_{\rm at}} \otimes H_f +
\lambda  V   ,
\end{equation}
where $\lambda \in \R$ is the coupling constant and $V = \phi(G)$, as a
semibounded
selfadjoint operator on the domain  $\dom(H(0))$. Let
$$
E(\lambda) = \inf \sigma(H(\lambda)) .
$$
Below we shall make the following assumption
\begin{hypothesis}   \label{hyp:0}
  There exists a positive constant $\lambda_0$ such that for all
  $\lambda \in [0,\lambda_0]$ the number  $E(\lambda)$ is a simple  eigenvalue of $H(\lambda)$
  with eigenvector  $\psi(\lambda) \in \HH$.
\end{hypothesis}

\begin{remark} \label{rem:exofgs}
We note that  the  existence of ground states has  been verified in many
cases \cite{Fro74,Spo98,BacFroSig99,Ger00,GriLieLos01}.
In particular, it has been shown in \cite{Ger00} that
Hypothesis \ref{hyp:0} holds if $H_{\rm at}$ has compact resolvent and
the coupling function satisfies
\begin{equation} \label{eq:assonG2}
\int \| G(k) \|^2(1+|k|^{-2}) dk < \infty .
\end{equation}
\end{remark}

We will outline in the next  section, that
if one  formally expands the eigenvalue equation for  the ground state in powers of the
coupling constant $\lambda$ and inductively solves    for the  expansion coefficients of the
ground state energy   one  obtains the recursion relation   \eqref{eq:E_nee1def}, below.
One can show that these expansion coefficients are indeed finite,
which is  the content of the next theorem. To formulate it we introduce the following notations.
We write
$$
H_0 = H(0)
$$
and
$$\psi_0 = \varphi_{\rm at} \otimes \Omega, $$
and denote by  $P_0$ the projection onto $\psi_0$ and let $\bar{P}_0 = 1 - P_0$.
Let  $P_\Omega$ denote the orthogonal projection in $\FF$ onto $\Omega$.
Then we can write
\begin{equation} \label{decompofproj}
\bar{P}_0 = P_{\rm at} \otimes \bar{P}_\Omega + \bar{P}_{\rm at} \otimes 1_\FF ,
\end{equation}
where $\bar{P}_\Omega = 1_\FF - P_\Omega$ and $\bar{P}_{\rm at} = 1_{\HH_{\rm at}} - {P}_{\rm at}$.

\begin{theorem} \label{thm:mainenergy} Suppose that \eqref{eq:assonG2}  holds.
 Then there exists a
  unique sequence $(E_n)_{n \in \N}$ in $\R$ such that
  \begin{align} \label{eq:E_0ee1}
    E_0 &= E_{\rm at} \\
    E_1 &=  \inn{ \psi_0 , V \psi_0 } \label{eq:E_nee1-2}    \\
    E_n &= \lim_{\eta \downarrow 0} E_{n}(\eta ) , \quad n \geq
          2  , \label{eq:E_nee1-3}
  \end{align}
  where
  \begin{align}\label{eq:E_nee1def}
     &E_n(\eta) :=   \\ &\;\;\sum_{k=2}^{n}
     \sum_{\substack{j_1 + \cdots + j_k = n \\ j_s \geq 1 }} \!\!\!\!
    \langle \psi_0  , ( \delta_{1 j_1 }  V - E_{j_1} )
    \prod_{s=2}^k  \left\{ (E_0 - \eta - H_0)^{-1} \bar{P}_0   (  \delta_{1 j_s }  V - E_{j_s}   ) \right\} \psi_0  \rangle    \nonumber
  \end{align}
  In particular the limit on the right hand side of
  \eqref{eq:E_nee1def} exists and is a finite number. The sequence
  $(E_n)_{n \in \N}$ can be defined inductively using
  \eqref{eq:E_0ee1}--\eqref{eq:E_nee1-3}.
\end{theorem}

\begin{remark} \label{rem:div} We note that the positive number $\eta$ appearing in \eqref{eq:E_nee1def} serves as
a regularization. The theorem states that the limit exists as the regularization is removed.
We note that this is not obvious, as some of the individual terms  on
the right hand side of \eqref{eq:E_nee1def}
diverge. This can be illustrated as follows. Consider for $n=2m$ the summand where  $j_s = 1$  for all $s$.
Inserting   \eqref{decompofproj} and  $a^*(G) + a(G)$ for $V$, multiplying out the resulting
expression, using  Wicks theorem and the so called pull through formula \cite[Appendix A]{BacFroSig98-1}
one obtains various terms.
One of them being
\begin{align}
 (-1)^{n-1}\!\! \int \! & dk_1 \cdots d k_m
 \langle  \varphi_{\rm at} ,  G^*(k_1) \frac{ P_{\rm at}  }{|k_1| + \eta}
 \colorboxed{red}{G^*(k_2)  \frac{P_{\rm at}}{|k_1|+|k_2| + \eta} G(k_2)}
 \frac{P_{\rm at}}{|k_1| + \eta}  \nonumber \\
 &\; \cdots\,
  G^*(k_m)  \frac{P_{\rm at}}{|k_1|+|k_m| + \eta} G(k_m) \frac{P_{\rm at} }{|k_1| + \eta}   G(k_1)  \varphi_{\rm at}   \rangle    , \label{eq:remarkex}
\end{align}
which is obtained by contracting the first and the last entry of the interaction and
contracting the remaining  nearest neighbor pairs. This can be symbolically pictured  as follows
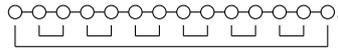
\begin{equation*}
 \contraction{1-1-1-1-1-1-1-1-1-1-1-1-1-1}{(1,14)_2,(2,3),(4,5),(6,7),(8,9),(10,11),(12,13)} .
   \end{equation*}
If $\eta \downarrow 0$ the integral over $k_1$ may  become divergent for large $m$. This is the case, for example,  if
$  \int dk |k|^{-m}  \| P_{\rm at} G(k) P_{\rm at} \|^2$  diverges for $m$ sufficiently large.
The convergence of   \eqref{eq:E_nee1def} can be restored using  cancellations originating from the  energy subtractions present  in the same  formula. To
illustrate this, consider the  summand where  $j_1 = 1$, $j_2=2$ and $j_3=\cdots=j_{n-1}=1$. As before one obtains
   various terms with one of them being the same as   \eqref{eq:remarkex} except for  the expression in the box which is replaced by $E_2 P_{\rm at}$.
   Thus adding these two terms  one can factor out
\begin{align}
&\int dk_2 P_{\rm at} G^*(k_2)  \frac{1}{|k_1|+|k_2| + \eta } G(k_2) P_{\rm at} +  E_2 P_{\rm at} \nonumber \\
&\qquad =  \int dk_2  P_{\rm at} G^*(k_2) \left( \frac{1}{|k_1|+|k_2| + \eta } -  \frac{1}{|k_2|} \right)  G(k_2) P_{\rm at} \nonumber \\
&\qquad = - (|k_1|+\eta)  \int dk_2P_{\rm at}  G^*(k_2)  \frac{1}{(|k_1|+|k_2| + \eta )|k_2|}  G(k_2)P_{\rm at} , \label{eq:exsingcanc}
\end{align}
where we used again  \eqref{eq:E_nee1def} to calculate $E_2$.
One sees that replacing the expression in the box in \eqref{eq:remarkex} by  \eqref{eq:exsingcanc}  remedies the
singularity  $k_1 \to 0$. To prove  Theorem  \ref{thm:mainenergy} we  will
show that similar cancellations can be carried out at every order.
\end{remark}

Once one has established the finiteness of the expansion coefficients of the ground state energy, we will
show that this yields an asymptotic expansion of the ground state energy. This is the content of
the following theorem.

\begin{theorem} \label{thm:groundstatee} Suppose \eqref{eq:assonG2} and Hypothesis~\ref{hyp:0} holds.
Then the  sequence
  $(E_n)_{n \in \N}$ defined in Theorem  \ref{thm:mainenergy}  yields an asymptotic
  expansion of the ground state energy, i.e.,
  \begin{align*}
    \lim_{\lambda \downarrow 0} \lambda^{-n} \left( E(\lambda)  - \sum_{k=0}^n  E_k \lambda^k \right)   &=  0 .
  \end{align*}
\end{theorem}

\begin{remark}
We want to note that if we would have the infrared condition
$
\int \| G(k) \|^2(1+|k|^{-2-\mu}) dk < \infty $,
for some $\mu > 0$, which is slightly stronger  than  \eqref{eq:assonG2},
then it would follow from \cite{GriHas09} that one has analyticity.
Moreover, there are couplings  with    \eqref{eq:assonG2}  where
 additional  symmetries may    cancel
 infrared divergencies such that the ground state energy is analytic \cite{HasHer11-2,HasHer11-3}.
\end{remark}

\begin{remark}
Note that in view of Remark \ref{rem:exofgs}   Hypothesis~\ref{hyp:0} is not a restrictive assumption.
And in many situations follows already from Inequality \eqref{eq:assonG2}.
\end{remark}

In the remaining parts of the paper we provide proofs of the above results and
furthermore we also  show the finiteness of the  expansion coefficients
for the ground state.

\section{Asymptotic Perturbation Theory}
\label{sec:pertheory}

In this section we derive formulas for the expansion
coefficients of the ground state and its energy.
Moreover  we show that provided  these coefficients are finite
up to some order, say  $n$,
 and a continuity assumption for  the ground state holds,
 then the ground state energy has an asymptotic
 expansion up to order $n$.
 We shall derive this result with two different methods.
 The first method in Subsection~\ref{subsec:formexp}  uses formal expansions and
 the comparison of coefficients combined with
 an analytic estimate.  The second method outlined  in
  Subsection~\ref{subsec:resfesh}  is based on  a  Feshbach type argument
  together with a  resolvent expansion.

We state our results for more general operators than introduced
in the previous section. Nevertheless we will use the same symbols
as in the previous section.
Let  $V$ and  $H_0$ be selfadjoint operators in a Hilbert space $\HH$.
To   prove our results we will use  the following assumption.

\begin{hypothesis}
  \label{hyp:1}
   The operator
    $H_0$ is bounded from below and   $V$ is $H_0$-bounded.
  There exists a positive constant $\lambda_0$ such that for all
  $\lambda \in [0,\lambda_0]$ there exists a simple  eigenvalue
  $
  E(\lambda )
  $
   of $$H(\lambda) = H_0 + \lambda V $$ with eigenvector $\psi(\lambda)$.
  Moreover,
  \begin{equation}
    \tag{H}
    \begin{split}
      \lim_{\lambda \to 0} \psi(\lambda) = \psi(0) \neq 0 , \quad
      \lim_{\lambda \to 0} E(\lambda) = E(0)
    \end{split}
  \end{equation}
  and
    \begin{equation}
    \tag{N}
    \langle \psi(0),\psi(\lambda) \rangle = 1
  \end{equation}
for all
  $\lambda \in [0,\lambda_0]$.
\end{hypothesis}
We note that (N)  can always be achieved  using a  suitable
normalization, possibly making the positive number   $\lambda_0$  smaller.
For notational convenience we shall write
  \begin{equation*}
  E_0 = E(0) , \quad   \psi_0 = \psi(0)     .
  \end{equation*}
Let $P_0$ denote the projection onto the kernel of $H_0 - E_0$ and let $\bar{P}_0 = 1 - P_0$.

\subsection{Expansion Method}
\label{subsec:formexp}

The idea behind the expansion method is to  expand the eigenvalue
equation in a formal power series and  equating coefficients.
This will lead to  Eq.   \eqref{eq:recurs}.
In Lemma   \ref{lem:abstractmain} we show that provided one
has a solution of  \eqref{eq:recurs}  up to some order $n$,
then the ground state energy has an asymptotic expansion up to the same order,
provided Hypothesis   \ref{hyp:1} holds.
In  Lemma \ref{lem:directform0} we inductively solve  \eqref{eq:recurs},
and in Lemma \ref{lem:directform} we give an explicit formula
for the inductive solution. We note that a similar result has been obtained in \cite{Ara14}.
However  in contrast to the result  in \cite{Ara14} we have less restrictive
assumptions.

\begin{lemma}
  \label{lem:abstractmain}
  Suppose
  Hypothesis~\ref{hyp:1} holds.  Let $n \in \N$ and suppose there exist
  $\psi_{1} ,\dotsc,\psi_{n} \in \bar{P}_0 \HH$ and
  $E_{1} ,\dotsc, E_{n} \in \C$ such that for all $m \in \N$ with
  $m \leq n$ we have
  \begin{equation} \label{eq:recurs} H_0 \psi_{m} + V \psi_{m-1} =
    \sum_{k=0}^m E_{k} \psi_{m-k} .
  \end{equation}
  Then  for all $ m \in \{1,\dotsc,n\}$ we have  that
  \begin{align}
    \lim_{\lambda \downarrow 0} \lambda^{-m} \left( E(\lambda)  - \sum_{k=0}^m  E_{k} \lambda^k \right)   &=  0  ,  \label{eq:asympgsab1} \\
    \lim_{\lambda \downarrow 0} \lambda^{-m}
    \langle  \psi_0 , V ( \psi(\lambda)  - \sum_{k=0}^m  \psi_{k} \lambda^k )   \rangle &=  0 . \label{eq:asympgsab2}
  \end{align}
\end{lemma}

First observe that \eqref{eq:recurs} implies that for all $m \leq n$
we have
\begin{align*}
  \inn{ \psi_{0} , V \psi_{m-1} } = E_{m} .
\end{align*}

\begin{proof}
  Proof by induction in $n$.  We define for
  $\lambda \in ( 0 , \lambda_0 )$ the quantities
$$
e_n(\lambda) := \lambda^{-n} ( E(\lambda) - ( E_0 + \lambda E_{1} +
\lambda^2 E_{2} + \cdots + \lambda^n E_{n} ) )
$$
$$
\rho_n(\lambda) := \lambda^{-n} ( \psi(\lambda) - ( \psi_0 + \lambda
\psi_{1} + \lambda^2 \psi_{2} + \cdots + \lambda^n \psi_{n} ) )
.
$$
Equation \eqref{eq:asympgsab2} for $m=0$ is just Hypothesis~\ref{hyp:1}.
Thus it remains to  show the induction step.

The eigenvalue equation gives
\begin{equation}  \label{eq:changeproj}
\bar{P}_0 ( H(\lambda) - E(\lambda) ) \bar{P}_0 \psi(\lambda)  =  -  \bar{P}_0  V {P}_0 \psi(\lambda) .
\end{equation}

\proofparagraph{$n-1 \to n$:}
Suppose that \eqref{eq:recurs} holds for all $m \in \{1,\dotsc,n\}$.
By induction Hypothesis we know that
$\lambda E_{n} + \lambda e_n(\lambda) \to 0$ and
$\inn{ V \psi_0 , \lambda \psi_{n} + \lambda \rho_n(\lambda)} \to
0$.
From the eigenvalue equation we find
\begin{align*}
   ( H_0 + \lambda V )&\left[ \sum_{k=0}^n  \lambda^k \psi_{k} + \lambda^n \rho_n(\lambda) \right] \\
  &\qquad\qquad= \left(\sum_{k=0}^n \lambda^k E_{k} + \lambda^n e_n(\lambda) \right)
  \left[ \sum_{k=0}^n \lambda^k \psi_{k} + \lambda^n \rho_n(\lambda) \right] .
\end{align*}
By ordering according to powers of $\lambda$ we see from
\eqref{eq:recurs} that many terms vanish and
\begin{align}
  \lambda V \psi_{n} + &(H_0 + \lambda V) \rho_n(\lambda)\nonumber \\
  &=
  \rho_n(\lambda) E(\lambda)
  + e_n(\lambda)\sum_{k=0}^n \lambda^k \psi_{k}
  + \sum_{k=n+1}^{2n}\lambda^{k-n} \sum_{j=k-n}^n
   E_{j} \psi_{k-j} . \label{eq:iterror}
\end{align}
If one applies $P_0$ to equation \eqref{eq:iterror} one obtains
$$
\lambda P_0 V ( \psi_{n} + \rho_n(\lambda)) = e_n(\lambda) \psi_0 .
$$
By induction Hypothesis the left hand side tends to zero as
$\lambda \to 0$.  This shows that \eqref{eq:asympgsab1} holds for all
$m \in \{1,\dotsc,n\}$.
Solving for terms involving $\rho_n(\lambda)$ in \eqref{eq:iterror} we
arrive at
\begin{align*}
  ( H(\lambda)  - E(\lambda)  )  \rho_n(\lambda)
  = e_n(\lambda) \sum_{k=0}^n \lambda^k \psi_{k}
  + \sum_{k=n+1}^{2n}\lambda^{k-n} \sum_{j=k-n}^n
  E_{j} \psi_{k-j}
  -\lambda V \psi_{n}.
\end{align*}
Applying $\bar{P}_0$ to this equation and using that $P_0 \rho_n(\lambda) = 0$ we find
\begin{align*}
 \bar{P}_0 ( H(\lambda)  -& E(\lambda)  ) \bar{P}_0 \rho_n(\lambda) \\
  &= \bar{P}_0 \left( e_n(\lambda) \sum_{k=0}^n \lambda^k \psi_{k}
  + \sum_{k=n+1}^{2n}\lambda^{k-n} \sum_{j=k-n}^n
  E_{j} \psi_{k-j}
  -\lambda V \psi_{n} \right).
\end{align*}
Calculating the inner product with
$\psi(\lambda)$ and using \eqref{eq:changeproj}  we find
\begin{align*}
  \langle & \psi(\lambda),  {P}_0  V  \rho_n(\lambda) \rangle \\
  &\;= -
    \inn{ \bar{P}_0  \psi(\lambda) ,
    e_n(\lambda) \sum_{k=1}^n \lambda^{k-1} \psi_{k}
    + \sum_{k=n+1}^{2n}\lambda^{k-n-1} \sum_{j=k-n}^n
    E_{j} \psi_{k-j}
    - \bar{P}_0  V \psi_{n}   }.
\end{align*}
This and Hypothesis~\ref{hyp:1} imply that \eqref{eq:asympgsab2} holds
for all $m \in \{1,\dotsc,n\}$.
\end{proof}

Next we inductively solve  Equation \eqref{eq:recurs}.

\begin{lemma}(Inductive Formula) \label{lem:directform0} Let
  $n \in \N$ and suppose there exist
  $\psi_1 ,\dotsc,\psi_n \! \in \bar{P}_0 \HH$ and $E_1,\dotsc,E_n \in \C$
  such that for all $m \in \N$ with $m \leq n$ we have
  \begin{equation}
    \label{eq:recurs1}
    H_0 \psi_{m} + V \psi_{m-1} =
    \sum_{k=0}^m E_k \psi_{m-k}.
  \end{equation}
  Then defining
  \begin{equation} \label{eq:givenene}
	E_{n+1} := \inn{ \psi_0 ,  V \psi_{n}  }
  \end{equation}
  as well as
  \begin{equation} \label{eq:inducstate} \psi_{n+1} := (H_0 - E_0)^{-1}
    \bar{P}_0 \left( \sum_{k=1}^{n+1} E_k \psi_{n+1-k} - V \psi_{n}
    \right) ,
  \end{equation}
  provided
    \begin{equation} \label{rel:domain}
  \bar{P}_0 ( \sum_{k=0}^n E_{k +1} \psi_{n-k} - V \psi_{n} ) \in \dom\left( (H_0
  - E_0 )^{-1} \bar{P}_0  \right),
    \end{equation}
  we obtain  a solution of   \eqref{eq:recurs1} for  $m=n+1$.
\end{lemma}

We note that the assumption in \eqref{rel:domain} is less restrictive
than the one in \cite{Ara14}, which will turn out to be  crucial to obtain  the asymptotic
expansion of the ground state to arbitrary order.

\begin{proof}
  This follows  by insertion of  \eqref{eq:inducstate}  and  \eqref{eq:givenene} into    \eqref{eq:recurs1} for $m=n+1$.
\end{proof}

If we solve the recursive relation  of the previous lemma, we obtain the following formulas.

\begin{lemma}(Direct Formula) \label{lem:directform}  Let
  $n \in \N$ and suppose there exist
  $\psi_1 ,\dotsc,\psi_n \in \bar{P}_0 \HH$ and $E_1,\dotsc,E_n \in \C$
  such that  the following holds.
  We have  $E_1  = \langle \psi_0 , V \psi_0 \rangle$,    for all $m \in \N$  with $ 2 \leq m \leq n$ we have
  \begin{align}
    & E_m  = \nonumber \\ & - \sum_{k=2}^{m} \sum_{\substack{j_1 + \cdots + j_k = m \\ j_s \geq 1 }}
    \langle \psi_0  , ( E_{j_1} - \delta_{1 j_1 }  V )
    \prod_{s=2}^k  \left\{ (H_0 - E_0)^{-1} \bar{P}_0   ( E_{j_s} - \delta_{1 j_s }  V   ) \right\} \psi_0  \rangle  ,  \label{eq:E_n}
  \end{align}
  and for all $m \in \N$ with $m \leq n$  we have
   \begin{align}
    \label{eq:E_ndirect}
    \psi_m  & = \sum_{k=1}^m \sum_{\substack{j_1 + \cdots + j_k = m \\ j_s \geq 1 }}
    \prod_{s=1}^k  \left\{ (H_0 - E_0)^{-1} \bar{P}_0   ( E_{j_s} - \delta_{1 j_s }  V   ) \right\} \psi_0  ,
  \end{align}
assuming  that the expressions on the right hand side  of     \eqref{eq:E_n}  and  \eqref{eq:E_ndirect} exist in the sense of  Lemma \ref{lem:directform0}.
  Then for all $m \in \N$ with $m \leq n$ we have
$$
H_0 \psi_{m} + V \psi_{m-1} = \sum_{k=0}^{m} E_k \psi_{m-k} .
$$
\end{lemma}
\begin{proof} We prove this lemma by induction in $n$. The case $n=1$ follows from a straight
forward calculation. Suppose the claim holds for $n$. Then also  the assumption of Lemma \ref{lem:directform0}
holds. Thus  we define $E_{n+1}$ as in   \eqref{eq:givenene}
  \begin{align*}
  &E_{n+1}  := \,  \inn{ \psi_0 ,  V \psi_{n}  }  \\
       &\, = -\! \sum_{k=2}^{n+1}
       \sum_{\substack{j_1 + \cdots + j_k = n +1 \\ j_s \geq 1 }}
    \!\!\!\!\! \langle \psi_0  , ( E_{j_1} - \delta_{1 j_1 }  V )
     \!\prod_{s=2}^k  \! \left\{\! (H_0 - E_0)^{-1} \bar{P}_0   ( E_{j_s} - \delta_{1 j_s }  V   )\! \right\} \psi_0  \rangle ,
  \end{align*}
  where in the second line we used the assumption  \eqref{eq:E_ndirect}
  (and note that  $\langle \psi_0, E_{j}   \bar{P}_0  (  \ \cdot   \ ) \rangle = 0$).
  We define $\psi_{n+1}$ as in  \eqref{eq:inducstate}
  \begin{align*}
  \psi_{n+1}  := & \, (H_0 - E_0)^{-1}
    \bar{P}_0 \left( \sum_{j=1}^{n+1} ( E_{j} - \delta_{1 j } V )
      \psi_{n+1-j} \right) \\
       = &\, \sum_{k=1}^{n+1} \sum_{\substack{j_1 + \cdots + j_k = n+1 \\ j_s \geq 1 }}
    \prod_{s=1}^k  \left\{ (H_0 - E_0)^{-1} \bar{P}_0   ( E_{j_s} - \delta_{1 j_s }  V   ) \right\} \psi_0 ,
  \end{align*}
  where we wrote the first line with  slightly different notation than in \eqref{eq:inducstate}  and in the second line we
  used the  assumption  \eqref{eq:E_ndirect}. Now it follows from  Lemma \ref{lem:directform0}
  that the claim of the lemma holds also for $n+1$.
\end{proof}

\subsection{Resolvent Method}
\label{subsec:resfesh}

Here we use a Feshbach type or Schur complement argument together with a
resolvent expansion.  The proof of the Lemma in this subsection is
inspired by \cite{Ara14}.

\begin{lemma}
  \label{lem:feshbach}
  Suppose  that Hypothesis~\ref{hyp:1} holds.
  Assume that starting with $ K_0   := \frac{\bar{P}_0}{H_0 - E_0}$ and
   $ E_1  :=    \inn{ \psi_0 , V  \psi_0 }$, we can define recursively
   for $m \in \{1,\dotsc,n-2\}$
   \begin{align}
     K_{m} & := \sum_{j=1}^{m}  K_{j-1} (E_{m+1-j} - \delta_{jm}V)K_0,
             \label{eq:recurskresm}  \\
     E_{m+1}& :=  - \inn{ \psi_0 , V K_{m-1} V \psi_0 } , \nonumber
   \end{align}
   such that $\bar{P}_0V \psi_0 \in  \dom(K_l)$ for $l=0,\dotsc,n-2$.
  Then $E(\lambda)$ has an asymptotic expansion up to order $n$,
  i.e.,
  for all $m = 1,\dotsc,n$
  \begin{equation*}
    \lim_{\lambda \downarrow 0} \lambda^{-m} \left( E(\lambda)  - \sum_{k=0}^m  E_k \lambda^k \right)   =  0.
  \end{equation*}
\end{lemma}

\begin{remark}
The statement of
  Lemma~\ref{lem:feshbach} is equivalent to the  statements of
  Lemma~\ref{lem:directform} and Lemma~\ref{lem:abstractmain} combined.
  In particular, we may solve iteratively for $K_m$ and obtain the relation
    \begin{equation*}
    V K_{m-2}  V = \sum_{k=2}^{m} \sum_{\substack{j_1 + \cdots + j_k = m \\ j_s \geq 1 }}
    ( E_{j_s} - \delta_{1 j_s }  V )\prod_{s=2}^k  \left\{ (H_0 - E_0)^{-1} \bar{P}_0  ( E_{j_s} - \delta_{1 j_s }  V   ) \right\} .
  \end{equation*}
    Moreover, given $E(\lambda)$, we can recover $\psi(\lambda)$ by
  \begin{equation*}
    \psi(\lambda) = \psi_0 -\lambda    \bar{P}_0 [ \left. \bar{P}_0 ( H(\lambda) -
        E(\lambda)) \bar{P}_0 \right|_{\ran(\bar{P}_0)}]^{-1} \bar{P}_0             V P_0 \psi_0.
  \end{equation*}
\end{remark}

\begin{proof}
  The eigenvalue equation $H(\lambda)\psi(\lambda) =
  E(\lambda)\psi(\lambda)$ can be split into the equivalent system of
  equations
  \begin{subequations}
    \begin{align}
      \label{eq:EV:P}
      P_0 \big(\lambda V +E_0 - E(\lambda)\big) P_0 \psi(\lambda) + \lambda P_0 V \bar{P}_0
      \psi(\lambda) &= 0 \\
      \label{eq:EV:P_perp}
      \lambda \bar{P}_0 V P_0 \psi(\lambda) + \bar{P}_0\big(H(\lambda) - E(\lambda)\big)
      \bar{P}_0 \psi(\lambda) &= 0,
    \end{align}
  \end{subequations}
  by applying the projections $P_0$ and $\bar{P}_0$ respectively.  From
  \eqref{eq:EV:P} we learn that
  \begin{equation*}
    \frac{E(\lambda) - E_0}{\lambda}\langle \psi_0, P_0 \psi(\lambda) \rangle
    - \langle \psi_0,  V P_0\psi(\lambda) \rangle
    = \langle V \psi_0,  \bar{P}_0 \psi(\lambda) \rangle
    = o(1),
  \end{equation*}
  i.e.
  \begin{equation*}
    \frac{E(\lambda) - E_0}{\lambda}
    \xrightarrow{\lambda \to 0} \langle \psi_0,  V \psi_0 \rangle.
  \end{equation*}
  This shows the claim for $n=1$. We show the lemma by induction.
  Suppose the claim holds for $n$ and the assumptions of the lemma hold for $n+1$.
Then  the recursively defined functions
  \begin{align}
    E^{[0]}(\lambda) &:= E(\lambda) \nonumber \\
    E^{[k]}(\lambda) &:= \frac{E^{[k-1]}(\lambda) - E_{k-1}}{\lambda} \label{eq:ebrackk}
  \end{align}
  satisfy
  \begin{equation*}
    \lim_{\lambda \downarrow 0}E^{[k]}(\lambda) =  E_k, \qquad
    k=0,\dotsc, n.
  \end{equation*}
  We write the  part $\bar{P}_0 \psi(\lambda)$ as follows
  \begin{align*}
    \bar{P}_0 \psi(\lambda)
    &= \frac{\bar{P}_0}{H_0 - E_0} (H_0 - E_0) \bar{P}_0\psi(\lambda) \\
    &= \frac{\bar{P}_0}{H_0 - E_0} [H(\lambda)-E(\lambda) +
    (E(\lambda) - E_0  - \lambda  V)] \bar{P}_0\psi(\lambda).
  \end{align*}
  Equation \eqref{eq:EV:P_perp} implies
  \begin{equation}
    \label{eq:EV:P_perp:2}
    \bar{P}_0 \psi(\lambda) = \lambda \frac{\bar{P}_0}{H_0 - E_0}\left[- V P_0\psi(\lambda)
    +  (E^{[1]}(\lambda)-  V)\bar{P}_0\psi(\lambda)\right].
  \end{equation}
  Iterated insertion of
  \eqref{eq:EV:P_perp:2} into itself, terminating the expansion after we have
  reached order $\lambda^{n}$, this leads to the following claim.
$$
$$
{\bf Claim:}  We have for $k =1,\dotsc, n$
  \begin{equation}
    \label{eq:expansion:perp}
   P_0 V \bar{P}_0 \psi(\lambda) =   P_0 V  \sum_{j=1}^{k} -\lambda^j K_{j-1} VP_0
    \psi(\lambda)
    +  P_0 V \lambda^k R_{k}(\lambda) \bar{P}_0 \psi(\lambda),
  \end{equation}
  where  $R_k(\lambda)$ is defined  by
  \begin{equation}
    \label{eq:expansion:R}
    R_{k}(\lambda) := \sum_{j=1}^k K_{j-1} (E^{[k+1-j]}(\lambda) - \delta_{jk}V).
  \end{equation}
  $$
  $$
  (We note that expressions are well defined by  the assumption
  $\bar{P}_0 V \psi_0 \in  \dom(K_l)$).
  Let us now show the claim. Equation~\eqref{eq:expansion:perp} for $k=1$ is just
  Equation~\eqref{eq:EV:P_perp:2} multiplied by $P_0 V$.
  Assume that \eqref{eq:expansion:perp} is true for a specific $k \leq n-1$.
  In this case, we insert first the Definition \eqref{eq:expansion:R} and then
  Definition \eqref{eq:ebrackk}
  \begin{align*}
    P_0 & V  R_k(\lambda) \bar{P}_0\psi(\lambda) \\
    &=  P_0 V \sum_{j=1}^k K_{j-1} (E^{[k+1-j]}(\lambda) -
       \delta_{jk}V)\bar{P}_0\psi(\lambda)\\
    &=  P_0 V \sum_{j=1}^{k}\left(  K_{j-1} (E_{k+1-j} -
      \delta_{jk}V)\bar{P}_0\psi(\lambda)
      + \lambda K_{j-1} E^{[k+2-j]}(\lambda)\bar{P}_0\psi(\lambda) \right) .
  \end{align*}
  We now use \eqref{eq:EV:P_perp:2} for the first summand and obtain
  \begin{align*}
    &P_0  V R_k(\lambda) \bar{P}_0\psi(\lambda) \\
    &\, = \lambda  P_0 V  \sum_{j=1}^k \Bigl(K_{j-1} (E_{k+1-j} -
         \delta_{jk} V)K_0 \bigl( - VP_0\psi(\lambda) + (E^{[1]}(\lambda)-V)
       \bar{P}_0 \psi(\lambda) \bigr)\\
     & \qquad\qquad\qquad\quad
       + K_{j-1} E^{[k+2-j]}(\lambda)\bar{P}_0\psi(\lambda)\Bigr).
  \end{align*}
  Using \eqref{eq:recurskresm} we find
  \begin{align*}
    P_0 V R_k(\lambda) \bar{P}_0\psi(\lambda)
    &= \lambda  P_0 V   \Bigl(K_k \bigl( - VP_0\psi(\lambda) + (E^{[1]}(\lambda)-V)
       \bar{P}_0 \psi(\lambda) \bigr)\\
    &\qquad\qquad\qquad
       + \sum_{j=1}^k K_{j-1} E^{[k+2-j]}(\lambda)\bar{P}_0\psi(\lambda)\Bigr)\\
    &= -  \lambda   P_0 V K_{k}VP_0\psi(\lambda) \\
       & \qquad+ \lambda P_0 V \sum_{j=1}^{k+1} K_{j-1}
        \bigl(E^{[k+2-j]}(\lambda) - \delta_{j,k+1}V\bigr)\bar{P}_0\psi(\lambda).
  \end{align*}
  By \eqref{eq:expansion:R} this expression agrees with
  \eqref{eq:expansion:perp} with $k$ replaced by $k+1$.
  Inserting this expression into \eqref{eq:expansion:perp} with $k$ we obtain
  \eqref{eq:expansion:perp} with $k$ replaced by $k+1$.
  This shows the claim.

  Next we
  insert the claim for $k=n$ into \eqref{eq:EV:P} to conclude
  \begin{equation}
    \label{eq:EV:expansion}
    \left(P_0 (E^{[1]}(\lambda) - V) P_0 + \sum_{j=1}^n \lambda^j P_0
      V K_{j-1} V P_0 \right)\! P_0 \psi(\lambda)
    = \lambda^n P_0 V R_n(\lambda)\bar{P}_0\psi(\lambda).
  \end{equation}
  Taking the inner product with $\psi_0$ and using the induction
  hypothesis \eqref{eq:recurskresm}, we obtain
    \begin{align*}
    E^{[1]}(\lambda) \,-&\,E_1   - \sum_{j=1}^n \lambda^j E_{j+1} \\
    &= \lambda^n \inn{ \psi_0 , P_0 V R_n(\lambda)\bar{P}_0\psi(\lambda)} =
   \lambda^n \inn{R_n(\lambda)  V \psi_0 ,   \bar{P}_0\psi(\lambda)} .
  \end{align*}
  Dividing by $\lambda^{n}$ we find using \eqref{eq:ebrackk}
  $$
    \lambda^{-(n+1)} \left(
     E(\lambda)    - \sum_{j=0}^{n+1} \lambda^j E_{j}  \right)
    =
    \inn{R_n(\lambda)  V \psi_0 ,   \bar{P}_0\psi(\lambda)}  = o(1).
  $$
 This shows the claim of the lemma for $n+1$.
\end{proof}
\begin{remark}
  Note that \eqref{eq:EV:expansion} implies
  \begin{equation*}
    \left(E(\lambda) -  \sum_{j=0}^n \lambda^j H_j \right) P_0 \psi(\lambda)
    = o(\lambda^n),
  \end{equation*}
  for $H_1 := P_0 V P_0$ and $H_n := - P_0 V K_{n-2} V P_0$.
  This can be used for a degenerate perturbation theory, where each
  operator $H_j$ has to be diagonalized and the coefficients $E_j$
  can be chosen out of the eigenvalues of $H_j$.
\end{remark}

\section{Ground State Energy}
\label{sec:groundstatee}

The main goal of  this section  is  to  show  Theorem  \ref{thm:mainenergy}.
As a corollary
we will obtain a formula for  the energies in terms of
so called linked contractions and renormalized propagators
(Corollary~\ref{cor:energyformula}).
For notational convenience we introduce the usual bosonic creation operators $a^*(k)$ and annihilation operators $a(k)$
satisfying canonical commutation relations
$$
[ a(k) , a(k') ] = 0 , \quad   [ a^*(k) , a^*(k') ] = 0 , \quad    [ a(k) , a^*(k') ] = \delta(k-k' ) ,
$$
for all $k, k' \in \R^3$.
Using creation and annihilation operators  we can write
$$
a^*(G) = \int G(k) a^*(k) dk  , \quad a(G) = \int G^*(k) a(k) dk   .
$$
Since we do not yet know the values  of the energies $E_n$ (indeed at this stage we do not even
know their   existence),
we shall write  in their place  $\EE_n$.  At the
end we will inductively construct the  energies $E_n$ as the
value of a limit.

Below we outline
the organization of this section and give an overview of the proof.
In Subsection~\ref{subsec:necNotation4} we
introduce
 notation which will be used in subsequent subsections.
 In Subsection~\ref{sec:opprod} we  provide  in  Lemma~\ref{lem:energyform2}
an alternative notation for the  energy coefficients \eqref{eq:E_nee1def}  in terms
of expectation values of  operator valued functions
  $T_n$, $n\in\N$,  which will be  defined in  \eqref{eq:E_n:2david} as a sum of operator products. This alternative notation
  will turn out to be convenient in  keeping track of the energy subtractions.
In Subsection~\ref{sec:defofcont} we use a
generalized version of
Wick's theorem  to express
 the operator valued functions  $T_n$ as a sum of contracted operator products,
see  Lemma \ref{lem:genwick}.
In   Subsection~\ref{sec:algebraicform}
 we use this result to  prove that $T_n$ is equal to an
 expression involving so called linked Feynman graphs, $C_n$,
plus a  sum of products of so called renormalized
linked Feynman graphs $\widehat{C}_m := C_m - \mathcal{E}_m$, $m < n$, with resolvents
in between, see  Proposition~\ref{thm:algebraicenergies}. The energy subtraction in  $\widehat{C}_m$
will eventually be responsible for the cancellation of the singularity in the  resolvent, as was illustrated in an example at the beginning
of the paper in Remark  \ref{rem:div}.
To obtain Proposition~\ref{thm:algebraicenergies} itself we start  with the  expression for $T_n$,  given in  Lemma  \ref{lem:genwick}.
We  separate  the sum over contractions   into   connected and  disconnected contractions, see    \eqref{eq:divideT}.
Then we use several involved algebraic reformulations to write the sum over disconnected contractions as a sum of products of connected
contractions.  Each of these
  connected expressions will come with  an  energy subtraction, as one may see  in Equation \eqref{eq:E_n:233db}.
  After  we have  proven  Proposition~\ref{thm:algebraicenergies}, it remains to show  that indeed the renormalized linked Feynman graphs $\widehat{C}_m$ cancel the singularity of the  resolvent.
  To this end, we first isolate the singular part of the  resolvent, by projection onto the space spanned by the atomic ground state $\varphi_{\rm at}$,
see   \eqref{eq:decofresolv}. Starting from Proposition~\ref{thm:algebraicenergies} we  then use elementary algebraic identities to rewrite   the operator valued functions $T_n$ in terms of the singular part of
the  resolvent. The resulting identity is stated in Proposition~\ref{thm:algebraicenergy}.
In  Subsection~\ref{subsec:estimateandProofThm1}   we finally prove
Theorem~\ref{thm:mainenergy}. To this end we use the identity of
Proposition~\ref{thm:algebraicenergy} for the operator valued functions $T_n$ and show, using an induction argument, that the renormalized linked Feynman graphs cancel the
singularity of the resolvent at each order.  The idea behind this induction argument
is explained in Remark \ref{rem:induction} at the beginning of   Subsection~\ref{subsec:estimateandProofThm1}.

\subsection{Graph functions  and Substitutions}
\label{subsec:necNotation4}

In this subsection we introduce notation which  will later be needed.

  \begin{definition}
    Let $G=(V,E)$ be a graph. Then a graph $G_1=(V_1,E_1)$ is called a
    {\bf subgraph} of $G$ and we write $G_1 \subset G$,   if $V_1 \subset V$ and $E_1 \subset E$.
For a subset $V_1 \subset V$ we define the
  {\bf restricted graph}
$$
G|_{V_1} := (V_1, \{ e \in E : e \subset V_1 \} ) .
$$
We define the union of  two graphs $G_1 = (V_1, E_1)$ and $G_2 = (V_2, E_2)$  by
$$
G_1 \cup G_2 := (V_1 \cup V_2 , E_1 \cup E_2 ) .
$$
  \end{definition}

For a subset $X \subset \Z$ we associated the graph  $G_X = ( X , E_X)$
with edges $E_X$ consisting of nearest neighbor pairs of $X$, that is
$$E_X := \{ \{ x , r_X(x)  \} : x \in X \setminus \max X \}, $$ where
$r_X(x) := \min ( X \setminus (-\infty,x] )$ denotes the nearest neighbor
vertex of $x$ which lies to the right.  We will also consider
graphs with external lines, that is
$$
\bar{G}_X := (X,\bar{E}_X)
$$
with
$$
\bar{E}_X := E_X \cup \{ \{ - \infty , \min X \} , \{ \max X , \infty
\} \} .
$$
In this subsection let   $\VV$ and $\RR$ denote two  sets. Later  we will refer to elements in
$\VV$ as interactions  and to elements in   $\RR$ as  resolvents or propagators.
\begin{definition}
    Let $X \subset \Z$.  For  $E = E_X$ and $E = \bar{E}_X$
    a function on $(X,E)$ of the form
  $$
  ((V_x)_{x \in X} , ( F_e )_{e \in E} ) ,
  $$
  where  $V_x \in \VV$ for every $x \in X$  and $F_e \in \RR$  for
  every
  $e \in E$ is called     a   $(\VV,\RR)${\bf -valued graph
  function on } $X$ and a   $(\VV,\RR)${\bf -valued graph
  function on } $X$
  {\bf with external lines}, respectively.
\end{definition}

\begin{example}
  We can write a graph function on $\{1, 2, 3, 4, 5\}$ symbolically as
  \begin{equation*}
    \node{1}{V_1}
    \edge{F_{\{1,2\}}}
    \node{2}{V_2}\edge{F_{\{2,3\}}}
    \node{3}{V_3}\edge{F_{\{3,4\}}}
    \node{4}{V_4}\edge{F_{\{4,5\}}}
    \node{5}{V_5}
  \end{equation*}
\end{example}

Now we introduce a so called substitution operation, which
substitutes  a piece of a graph function by a simpler expression.
This will later be used to express so called  renormalized
Feynmann graphs.

\begin{definition} \label{def:subst}  Let $X \subset \Z$, and let $K \in \RR$.
Let $\pi = ((V_x)_{x \in X} , ( F_e )_{e \in \bar{E}_X} )$ be
    a $(\VV,\RR)$-valued graph function with external lines on
    $X$. For $I \subset X$ with
    $G_{I} \subset G_X$ we define
    \begin{equation}
      \label{eq:substext}
      \subst_{\substack{I \to K }
      }(\pi) := ((V_x)_{x \in X \setminus I } , ( \tilde{F}_e )_{e \in
        \bar{E}_{X \setminus I} } ) ,
    \end{equation}
    where for $e \in \bar{E}_{X \setminus I}$
    \begin{equation*}
      \tilde{F}_e :=
      \begin{cases}
        {F}_{ \{ \min e , \min I \} } K
        {F}_{ \{ \max e , \max I \} } \, ,
        & e \notin \bar{E}_X  ,  \\
        {F}_e  \, , &  e \in \bar{E}_X
      \end{cases}
    \end{equation*}
\end{definition}

Note that \eqref{eq:substext} is again a $(\VV,\RR)$-valued graph
function on $X \setminus I$ with external lines, that is, for a
graph function with external lines we can substitute
any subgraph and we obtain again a graph function with external lines.
In Subsection   \ref{sec:opprod}   we show how Definition \ref{def:subst} can
 be naturally extend  to
 graph functions without external lines.

  \begin{example}
    Let $\pi$ denote the graph function of the previous example.
    Suppose $I = \{2,3\}$. Then we write the graph function
    $ \subst_{\substack{I \to K } }(\pi) $ symbolically as
  \begin{equation*}
    \node{1}{V_1}
    \edge{F_{\{1,2\}} K F_{\{3,4\}}}
    \node{4}{V_4}\edge{F_{\{4,5\}}}
    \node{5}{V_5}.
  \end{equation*}
\end{example}

The following lemma is a direct consequence of the definition.

\begin{lemma}(Commutativity) \label{lem:comm}  Let $X \subset \Z$ and   let $F_I, F_J \in
  \RR$.  Let $\pi$ be a $(\VV,\RR)$-valued graph function
  on $X$ with external lines.  For any disjoint subsets  $I,J $ of
  $X$ with  $G_{I} , G_J \subset G_X$, we have
  $$
  \subst_{\substack{I \to F_I } }(\pi) \subst_{\substack{J \to F_J }
  }(\pi)= \subst_{\substack{J \to F_J } }(\pi) \subst_{\substack{I \to
      F_I } }(\pi).
  $$
\end{lemma}

 Lemma \ref{lem:comm} justifies the use of the following notation.
Let   $X \subset \N$ and
 let a set  $\mathcal{I}$  of mutually disjoint subsets of $X$
be given such that $G_{I} \subset G_X$ for all $I \in \mathcal{I}$, and let
for each $I \in \mathcal{I}$  an element $F_I \in \RR$   be given. Then we write for any
 $(\VV,\RR)$-valued graph functions $\pi$   on $X$ with external
lines
  $$
  \subst_{\substack{I \to F_I \\ I \in \mathcal{I} } }(\pi) :=
  \prod_{I \in \mathcal{I} } \{ \subst_{\substack{I \to F_I } }
  \}(\pi) .
  $$

\subsection{Operator Products}
\label{sec:opprod}
In  this subsection we use the above notation to write  the   energy coefficients
in terms of  expectation values   of  operator products. To this end
we introduce the set of interactions and the set of propagators
suitable for the generalized spin boson model
\begin{align*}
   \VV_{\rm sb} & = \{ a^*(G) + a(G) : \, G \in \mathcal{L}(\R^d ;
   \mathcal{L}(\HH_{\rm at}), (1 + |k|^{-2} ) dk) \}  \\
   \RR_{\rm sb} & = \{ R : [0,\infty ) \to \mathcal{L}(\HH_{\rm at})
   \text{ piecewise   continuous}
   \} .
\end{align*}
If we are given a  $(\VV_{\rm sb},\RR_{\rm sb})$-valued graph function on $X$ with no
external lines, we can naturally extend it to a graph function with external lines
by  assigning  the identity operator in $\HH_{\rm at}$ to  each
external line. With this extension we can naturally extend every definition for
graph functions with external lines to such without external lines.
In particular we can extend Definition \ref{def:subst} to graph functions without
external lines.

For a finite set  $X \subset \Z$ and for   $\phi = ((V_x)_{x \in X} , ( F_e )_{e \in \bar{E}_X})$ a
   $(\VV_{\rm sb},\RR_{\rm sb})$-valued graph function on $X$ with
   external lines, we define the formal  operator product
  $$
  \Pi(\phi ) = F_{\{-\infty, \min X \}}(H_f) \!\!\!\! \prod_{ x \in X \setminus
    \max X} \!\!\!\! \left\{ V_x F_{ \{ x, r_X(x) \} }(H_f) \right\}\! V_{\max X}
  F_{\{\max X, \infty \}}(H_f) .
  $$
  Moreover, we define an energy shift
  $$
  \mathcal{T}_r(\phi) := ((V_x)_{x \in X} , ( F_e(\cdot + r  )_{e \in \bar{E}_X}) ,
  $$
  for $r \geq 0$.
Let us now define a special graph function, which we will eventually use
to write  the expansion coefficients of the energy.
   For $r \geq 0$ and $\eta \geq 0$ define
  \begin{align}\label{eq:defofresolvent}
  R(r,\eta) := \frac{ 1-P_{\rm at} \otimes 1_{ r = 0} }{E_0 - H_{\rm
      at} - r - \eta } .
  \end{align}
  We note that \eqref{eq:defofresolvent} is bounded if $\eta > 0$ or
  $r > 0$.
  The parameter $\eta$ serves as a regularization which we
  shall later remove.  The parameter $r$ will be needed to account for
  the additional terms arising from the pull-through formula.

   Let us now define the graph functions which we will use
   for our model.  We write
  $$N_n :=  \N \cap [1, n ]  ,$$
  and we define for $r , \eta \geq 0$
  \begin{align*}
  \pi_n(\eta) & := ((V)_{x \in
      N_n} ,  ( R( \cdot   , \eta) )_{e \in
      E_{N_n}} ) \\
  \pi_n(r,  \eta) &:= ((V)_{x \in
      N_n} ,  ( R( \cdot + r  , \eta) )_{e \in
      E_{N_n}} ) = \mathcal{T}_r \pi_n(\eta)  ,
  \end{align*}
 where $V \in \VV_{\rm sb}$ is the interaction of the spin boson model.

  For a  given  sequence $(\EE_n)_{n \in \N}$ in $\R$  we define the expression
  \begin{align}
    \label{eq:E_n:2david}
    T_n(r, \eta)
      &:=  \sum_{k=1}^n \sum_{\substack{ \{I_1,\dotsc,I_k\} \\ I_i \cap I_j = \emptyset,  \, I_j \neq \emptyset \\ G_{I_i} \subsetneq G_{N_n}   }}
    ( 1 \otimes P_\Omega )   \Pi (\subst_{\substack{I  \to -\EE_{|I|}\\ I \in \{I_1,\dotsc,I_k\} }   }(\pi_n( r,  \eta))) ( 1 \otimes P_\Omega )  ,
  \end{align}
  where the second sum is over all sets with $k$-elements, with
  elements being nonzero disjoint subsets of $N_n$ such that their
  associated graphs are subgraphs of $G_{N_n}$ (this condition is
  imposed to ensure that $I_k$ does not contain any holes). We shall
  adopt the convention that we view \eqref{eq:E_n:2david} as an
  operator restricted to the atomic Hilbert space.
We also introduce the expressions
$$
\widehat{T}_n(r, \eta) := T_n(r, \eta) - \EE_n ,
$$
which we will refer to as renormalized propagators.
Henceforth we shall write  $ P_\Omega$ instead of  $1 \otimes P_\Omega$,
In the following lemma we relate the energy formula \eqref{eq:E_nee1def} in Theorem \ref{thm:mainenergy}
to the  expressions defined in  \eqref{eq:E_n:2david}.

\begin{lemma} \label{lem:energyform2} Suppose $\eta > 0$. Then for any sequence $(\EE_n)_{n \in \N}$ of real numbers
we have for $n \geq 2$
  \begin{align*}
    &\inn{ \varphi_{\rm at} , T_{n}(0,\eta) \varphi_{\rm
    at} } =  \\  & \; \sum_{k=2}^{n} \sum_{\substack{j_1 + \cdots + j_k = n \\ j_s \geq 1 }}
    \langle \psi_0  , ( \delta_{1 j_1 }  V - \EE_{j_1} )
    \prod_{s=2}^k  \left\{ (E_0 - \eta - H_0)^{-1} \bar{P}_0   (  \delta_{1 j_s }  V - \EE_{j_s}   ) \right\} \psi_0  \rangle  .
  \end{align*}
\end{lemma}
\begin{proof} To see
  this, we identify each summand in the sum.  Consider the summand in
  \eqref{eq:E_n:2david} indexed by
  $\mathcal{I} = \{ I_1,\dotsc,I_l \}$.  We complement this set by
  sets consisting of elements of $N_n$ which are not contained in any
  of the sets in $\mathcal{I}$.  To this end we define
$$
\mathcal{J} := \{ \{ s \} : s \in N_n \text{ and } s \notin I , \
\forall I \in \mathcal{I} \} .
$$
Now we order the elements of
$\mathcal{S} := \mathcal{I} \cup \mathcal{J}$ in increasing order in
the sense that for all $s_1, s_2 \in \mathcal{S}$ we set
$$
s_1 < s_2 : \Leftrightarrow \text{every element of } s_2 \text{ is an
  upper bound of } s_1 .
$$
This defines a bijection $\varphi: N_{|\mathcal{S}|} \to \mathcal{S}$
preserving the order.
By construction we see that the summand in
\eqref{eq:E_n:2david} indexed by $\{I_1, \dotsc, I_l\}$ is equal to the
summand in \eqref{eq:E_nee1def}, which we obtain by choosing $k = |\mathcal{S}|$, the indices
  $j_s = |\varphi(s)|$ for $s =1,\dotsc,k$,    by choosing  $-\EE_1$ in case  $j_s=1$ and $\varphi(s) \in \mathcal{I}$,
  and by choosing  $V$ if $j_s=1$ and $\varphi(s) \notin \mathcal{I}$.
\end{proof}

  Finally we give an alternative formulation  for \eqref{eq:E_n:2david}, to shorten the
  notation in forthcoming proofs.
  For  $S \subset \Z$ we  say that a set of the form $\{ s \in S :  a \leq s \leq b \}$ for some  $a,b \in S$ is an interval
  of $S$.   For $M \subset \Z$ we define the
   set  $\mathcal{Q}(M)$   consisting of all collections  of disjoint  nonempty intervals  of  $M$, i.e.,
  \begin{align*}
    \mathcal{Q}(M) & := \{ \mathcal{I} \subset \mathcal{P}(M)   :  \ \forall   I , J \in \mathcal{I} \text{ we have  }  I \cap J = \emptyset  , \,
     \\
    & \qquad \quad  \forall I  \in \mathcal{I} \text{ the set  } I  \text{ is a nonempty  interval of }  M   \} .
  \end{align*}

  Then we can rewrite \eqref{eq:E_n:2david} as
  \begin{align}
    \label{eq:E_n:2}
    T_n(r, \eta)
      &=  \sum_{\substack{ \mathcal{I} \in  \mathcal{Q}(N_n) \\ N_n \notin \mathcal{I}  }}
    P_\Omega  \Pi (\subst_{\substack{I  \to -\EE_{|I|}\\ I \in \mathcal{I} }   }(\pi_n(r,   \eta))) P_\Omega ,
  \end{align}
  and for the renormalized expression
  \begin{align} \label{eq:E_n:2:renom}
       \widehat{T}_n(r, \eta)
      &=  \sum_{\substack{ \mathcal{I} \in  \mathcal{Q}(N_n)  }}
    P_\Omega  \Pi (\subst_{\substack{I  \to -\EE_{|I|}\\ I \in \mathcal{I} }   }( \pi_n(r,  \eta))) P_\Omega .
  \end{align}

\subsection{Wicks Theorem and Contractions}

\label{sec:defofcont}

Now we use a generalized version of Wicks theorem to write
 \eqref{eq:E_n:2david} as a sum of so called contractions.  To this end we
introduce the following notation.

\begin{definition}
  Let $X$ be a finite set. A {\bf pair} of $X$ is a subset of $X$
  containing two elements. A {\bf pair partition} $P$ of $X$ is a partition
  of $X$ consisting of pairs of $X$, i.e., $|X|$ is even and we have
$$
P = \{ p_1 , p_2,\dotsc,p_{\frac{|X|}{2}} \}
$$
where $p_j$ is a pair of $X$, $p_i \cap p_j = \emptyset $ if
$i \neq j$, and $\bigcup_{p \in P} p = X$.
A {\bf pairing} of $X$ is a
pair partition of a subset of $X$.
\end{definition}

\begin{definition}\label{def:defofcont}
Let  $X \subset \Z$ be finite, let $P$ be a pairing of $X$, and let
$\phi = ((a(G_x)+a^*(G_x))_{x \in X} , ( F_e )_{e \in \bar{E}_X})$ be a
$(\VV_{\rm sb},\RR_{\rm sb})$-valued graph function with external lines on $X$.
Then the {\bf contraction}  of $\phi$ with respect to $P$ is defined by
\begin{align*}
  \CC_P(  \phi )(r) &:= \prod_{j \in X }  \left\{ \int dk_j  \right\}
  \delta_{P}(k) F_{\{-\infty,\min X \}}(r) \\
  & \qquad\quad \prod_{j \in X  \setminus \{ \max X \}}
	\left\{   G^\sharp_{j,P}(k_j) F_{\{j,r_X(j) \} }
		(|K_{\{j,r_X(j)\}}|_P + r )   \right\} \\
  &\qquad\qquad\qquad\qquad\quad\;\times
  G^\sharp_{\max X, P}(k_{\max X}) F_{\{\max X ,\infty \}}(r) ,
\end{align*}
where $r \geq 0 $ and
\begin{align*}
  G^\sharp_{j,P}   & :=
              \begin{cases}
                G_j^*     & , \exists p \in P ,  \  j = \min p  \\
                G_j & , \exists p \in P , \ j = \max
                      p
              \end{cases}
              \, ,
  \\
  \delta_P(k) & := \prod_{p \in P} \delta(k_{\min p} - k_{\max p} )\, , \\
  |K_e |_P & := \sum_{\substack{ p \in P \\ \max e \leq \max p \\ \min p
  \leq \min e }} | k_{\max p}|\, .
\end{align*}
We define
$$
  \CC_P^0(  \phi ) :=   \CC_P(  \phi )(0) .
$$
\end{definition}
We shall adopt the following conventions.  We write
$$
\CC( \phi ) := \sum_P \CC_P(\phi ) ,
$$
where the sum is over all pair partitions of $X$.
If $P$ is not a pair partition of
$X$ we set  for notational compactness
$$
\CC_P( \phi ) = 0 .
$$
We adopt analogous conventions for $\mathcal{C}^0$.

\begin{lemma}
Suppose the situation is as in   Definition  \ref{def:defofcont},
and that for all $e \in \bar{E}_X$ the functions $F_{e}(r)$ are uniformly bounded in  $r \geq 0$ .
Then $\mathcal{C}(\phi) \in \mathcal{R}_{\rm sb}$.
\end{lemma}
\begin{proof} This follows from the  dominated convergence theorem.
\end{proof}

\begin{remark} \label{rem:equivform} Let the situation be as in Definition \ref{def:defofcont}.
Then we have
$$
\CC_P(  \phi )(r) = \CC_P^0(  \mathcal{T}_r \phi ) .
$$
\end{remark}

We illustrate Definition \ref{def:defofcont} in the following example.

\begin{example}\label{ex:cont} Consider the set $N_4$ and
we consider the pair partition $P$ which is indicated by the lines below.
 
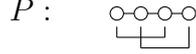
\begin{equation*}
\begin{array}{ll}
 P : \quad &    \contraction{1-1-1-1}{(1,3)_1,(2,4)_2}
   \end{array}
\end{equation*}
\begin{align*}
  \CC_P(  \phi )(r)
   =  \int dk_1 &\cdots dk_4 \Big\{ \delta(k_1-k_3) \delta(k_2-k_4)
		F_{\{-\infty,1\}}(r)  G^*(k_1) \\
     &\qquad \qquad \;\, F_{\{1,2\}}(|k_3|+r) G^*(k_2) F_{\{2,3\}}(|k_3| + |k_4|+ r) \\
     &\qquad \qquad \qquad\qquad\; G(k_3) F_{\{3,4\}}(|k_4| + r ) G(k_4) F_{\{4 ,\infty \}}(r)\Big\}.
\end{align*}
\end{example}

\begin{lemma}[Generalized Wick Theorem]\label{lem:wick}
  Let $X \subset \N$ be finite and let $\phi$ be a
  $(\VV_{\rm sb},\RR_{\rm sb})$-valued graph function on $G_X$ or
  $\bar{G}_X$. Then
  $$
  P_\Omega \Pi(\phi) P_\Omega = \CC^0(\phi) \otimes |\Omega \rangle \langle \Omega | .
  $$
\end{lemma}
The proof follows from the usual Wick theorem, leaving the operator
valued functions $G$ at their position and using the so called pull
through formula.  The pull through formula gives the commutation
relation between the free field energy and the creation or
annihilation operators.  For a detailed proof we refer the reader to
\cite{BacFroSig98-1}.
The following lemma is an  immediate consequence  of
\eqref{eq:E_n:2} resp. \eqref{eq:E_n:2:renom}
and the generalized Wick theorem (Lemma~\ref{lem:wick}).

\begin{lemma}\label{lem:genwick} We have
  \begin{align}
    \label{eq:E_n:22}
    T_n(r,\eta)
      &=     \sum_{\substack{\mathcal{I} \in  \mathcal{Q}(N_n) \\ N_n \notin \mathcal{I}   }}
    \mathcal{C}^0  (\subst_{\substack{I \to -\EE_{|I|}\\ I \in \mathcal{I }} }(\pi_n(r,\eta))), \\
    \widehat{T}_n(r,\eta) \label{eq:E_n:22:renorm}
      &=     \sum_{\substack{\mathcal{I} \in  \mathcal{Q}(N_n)    }}
        \mathcal{C}^0  (\subst_{\substack{I \to -\EE_{|I|}\\ I \in \mathcal{I }} }(\pi_n(r,\eta)))      .
  \end{align}
\end{lemma}

We shall need the following lemma which collects two algebraic relations of the contraction
operation.

\begin{lemma}
  \label{lem:algrelcont}
  Let  $X \subset \Z$ be finite, let $P$ be a pairing of $X$.
  \begin{enumerate}[label=\rm (\alph*)]
  \item \label{lem:algrelcont:a}
    For $j=1,2$ let
    $\phi_j = ((V_x)_{x \in X} , ( F_{j,e})_{e \in \bar{E}_X})$ be
    $(\VV_{\rm sb},\RR_{\rm sb})$-valued graph functions with external lines on $X$ and
    suppose for $e \in \bar{E}_X$  we are given numbers $t_{e,j} \in \C$.
    Then  the following multilinearity relation holds
    \begin{align*}
      \mathcal{C}_P\big((V_x)&_{x \in X} , ( \sum_{j_e=1,2} t_{e,j_e} F_{j_e,e})_{e \in \bar{E}_X} \big) \\
      &\quad=\prod_{e \in \bar{E}_X} \left\{  \sum_{j_e=1,2} t_{e,j_e} \right\}
      \mathcal{C}_P((V_x)_{x \in X} , (  F_{j_e,e})_{e \in \bar{E}_X}) .
    \end{align*}
  \item \label{lem:algrelcont:b} Suppose we are given disjoint sets $X_{\rm l},  X_{\rm r}   \subset X$ such that their union equals $X$
    and
    $$
    \max X_{\rm l}   < \min X_{\rm r} .
    $$
    Furthermore, assume that $P = P_{\rm l} \cup P_{\rm r}$ where $P_{\rm l}$ is a pair partition of $X_{\rm l}$ and
    $P_{\rm r}$ is a pair partition of $ X_{\rm r}$. Then for any
    $(\VV_{\rm sb},\RR_{\rm sb})$-valued graph function, $\phi =
    ((V_x)_{x \in X} , ( F_{e})_{e \in {E}_X})$,  we have the
    \begin{align*}
      \mathcal{C}_P(\phi) =
      \mathcal{C}_{P_{\rm l}}(\phi_{\rm l}  )   F_{\{ \max X_l, \min X_r \}}
      \mathcal{C}_{P_{\rm l}}(\phi_{\rm r}  )  ,
    \end{align*}
    where we defined
    \begin{align*}
      \phi_{\rm l} := ((V_x)_{x \in X_{\rm l}} , ( F_{e})_{e \in {E}_{X_{\rm l}}})  )  , \quad
      \phi_{\rm r} := ((V_x)_{x \in X_{\rm r}} , ( F_{e})_{e \in {E}_{X_{\rm r}}})  )  .
    \end{align*}
  \item \label{lem:algrelcont:c} Suppose we are given disjoint sets $X_{\rm l}, X_{\rm m}, X_{\rm r}   \subset X$ such that their union equals $X$
    and
    $$
    \max X_{\rm l} < \min X_{\rm m}  < \max X_{\rm m}  < \min X_{\rm r} .
    $$
    Furthermore, assume that $P = P_{\rm m} \cup P_{\rm b}$ where $P_{\rm m}$ is a pair partition of $X_{\rm m}$ and
    $P_{\rm b}$ is a pair partition of $X_{\rm l} \cup X_{\rm r}$. Then for,  $\phi = ((V_x)_{x \in X} , ( F_{e})_{e \in \bar{E}_X})$, any
    $(\VV_{\rm sb},\RR_{\rm sb})$-valued graph functions with external lines on $X$  we have the following substitution relation
    \begin{align*}
      \mathcal{C}_P(\phi) = \mathcal{C}_{P_{\rm b}}(
      \subst_{
      \substack{
      X_{\rm m}  \to  \tilde{\phi}
      }
      }
      (\phi)  ) ,
    \end{align*}
    where
    $$
    \tilde{\phi} :=   \mathcal{C}_{P_{\rm m}}((V_x)_{x \in X_{\rm m}} , ( F_{e})_{e \in {E}_{X_m}})  .
    $$
  \end{enumerate}
\end{lemma}
\begin{proof}
  \ref{lem:algrelcont:a}. This follows  from   the bilinearity of the operator product and the linearity
  of the integral.
  Statements \ref{lem:algrelcont:b} and \ref{lem:algrelcont:c} follow from Fubinis Theorem.
\end{proof}

\subsection{Feynman Graphs, Renormalization}
\label{sec:algebraicform}

In this subsection we want to evaluate  the sum over all  contractions
in  \eqref{eq:E_n:22}.
To this end we will show that we can write  \eqref{eq:E_n:22} as a sum of
so called linked contractions over so called renormalized propagators.
 We shall use the notation, that for a set
$\mathcal{A}$ of sets we write $\bigcup \mathcal{A} := \bigcup_{A \in \mathcal{A}} A$.

\begin{definition} Let $X \subset \Z$.  We call two distinct elements
  $p_1$ and $p_2$ of a pairing of $X$ {\bf linked} if one element of
  $p_1$ lies between the elements of $p_2$ and one of the elements of
  $p_2$ lies between the elements of $p_1$, i.e.,
$$
p_1 \cap [\min p_2 , \max p_2 ] \neq \emptyset \quad \text{ and } \quad p_2
\cap [\min p_1 , \max p_1 ] \neq \emptyset .
$$
For a pairing $P$ of $X$, we call the mapping
$$
\gamma : \{ 0 ,\dotsc, l \} \to P ,
$$
with $l \in \N$, a {\bf linked path} in $P$ from $\gamma(0)$ to
$\gamma(l)$ of {\bf length}  $l$ if $\gamma(i)$ and $\gamma(i+1)$ are linked
for all $i=0,\dotsc,l-1$.  A pairing $P$ of $X$ is called {\bf linked}
if for any two elements $p_1, p_2 \in P$ there exists a linked path in
$P$ from $p_1$ to $p_2$.
The property that there exists a linked path
between two pairings is an equivalence relation on $P$, and we call
the equivalence classes {\bf linked  components} of $P$.
We say that $P$ {\bf links two elements}   $x,  y$ of $X$ if $P$ has
a linked component $P_0$ such that $x,y \in \bigcup P_0$.
\end{definition}

\begin{example} Consider  $N_{14}$.
Then the pairing   $P = \{\{1,5\},\{4,13\},\{12,14\} \}$ is linked
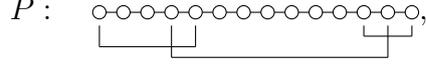
\begin{equation*}
P : \quad     \contraction{1-1-1-1-1-1-1-1-1-1-1-1-1-1}{(1,5)_2,(4,13)_3,(12,14)} ,
        \end{equation*}
  and the pairing $Q =    \{ \{1,5\},\{4,13\},\{6,8\},\{7,11\},\{12,14\}\}$,
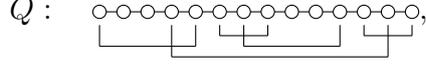
  \begin{equation*}
    Q : \quad \contraction{1-1-1-1-1-1-1-1-1-1-1-1-1-1}{(1,5)_2,(4,13)_3,(6,
      8),(7,11)_2,(12,14)} ,
  \end{equation*}
  can be written as the union of  the linked components $P$ and $\{  \{6,8\}, \{7, 11\} \}$.
\end{example}

Next we consider the set of non-paired elements. Specifically,
for any
 pairing $P$ of $N_n$   we define
$\mathcal{I}_P$  to be coarsest  partition of the
 the set of all partitions of    $N_n \setminus \bigcup P$
into intervals of $N_n$.  This definition is illustrated in the next example.

\begin{example} Consider  $N_{14}$  and $P = \{\{1,5\},\{4,13\},\{12,14\} \}$.
\begin{equation*}
	P : \quad     \contraction{1-1-1-1-1-1-1-1-1-1-1-1-1-1}{(1,5)_2,(4,13)_3,(12,14)}
\end{equation*}
Then
        $
        \mathcal{I}_P = \{ \{2,3\} , \{6,7,8,9,10,11\} \} .
        $
\end{example}

\begin{remark} Let us characterize $\mathcal{I}_P$ in different terms.
The set $\mathcal{I}_P$ is  the unique  partition of  $N_n \setminus \bigcup P$ such that
$
\bigcup_{I \in \mathcal{I}_P}  G_{I} = G_{N_n} |_{N_n \setminus \cup P}
$.
\end{remark}

We define
\begin{align}\label{eq:defofc}
  C_n(r, \eta) := &  \sum_{\substack{ P \text{ pairing of } N_n  \\ \{ 1,n  \}  \subset  \cup P    \\  P \text{ linked }     } }
  \mathcal{C}_{P}( \subst_{\substack{I  \to \widehat{T}_{|I|}(\cdot,\eta)    \\ I \in \mathcal{I}_P }} (
  \pi_n(\eta))(r) \\
   =  & \sum_{\substack{ P \text{ pairing of } N_n  \\ \{ 1,n  \}  \subset  \cup P    \\  P \text{ linked }     } }
  \mathcal{C}_{P}^{0}( \subst_{\substack{I  \to \widehat{T}_{|I|}(\cdot + r,\eta)    \\ I \in \mathcal{I}_P }} (
  \pi_n(r,\eta)) , \nonumber
\end{align}
where we  refer to the summand on the right hand
side as a linked  Feynman graph with renormalized propagators. We define
 \begin{equation} \label{defofchat}
 \widehat{C}_n(r,\eta) := C_n(r,\eta) - \EE_n .
\end{equation}

\begin{proposition}\label{thm:algebraicenergies}
  We have
  \begin{align*}
    T_n(r,\eta) & = C_n(r,\eta)  +  \sum_{k=2}^n \sum_{\substack{j_1+\ldots + j_k = n\\ j_i \geq
    1}} \left[ \prod_{i=1}^{k-1}\bigl( \widehat{C}_{j_i}(r,\eta)  R(r,\eta) \bigr)
    \right] \widehat{C}_{j_k}(r,\eta).
  \end{align*}
\end{proposition}

For the proof we will introduce the notion  of connected pairings,
which is similar to the notation of linked pairings, but not the same.

\begin{definition} Let $X \subset \Z$.  For a pairing $P$ of $X$ we call
the set
$$
S(P) = \bigcup_{p \in P} [\min p,\max p]
$$
the span of $P$.
We call two distinct elements
  $p_1$ and $p_2$ of a pairing of $X$ {\bf connected} if one element of
  $p_1$ lies between the elements of $p_2$ or  one of the elements of
  $p_2$ lies between the elements of $p_1$, i.e.,
$$
 [\min p_1 , \max p_1 ]   \cap [\min p_2 , \max p_2 ] \neq \emptyset .
$$
For a pairing $P$ of $X$, we call the mapping
$$
\gamma : \{ 0 ,\dotsc, l \} \to P ,
$$
with $l \in \N$, a {\bf connected path} in $P$ from $\gamma(0)$ to
$\gamma(l)$  if $\gamma(i)$ and $\gamma(i+1)$ are connected
for all $i=0,\dotsc,l-1$.  A pairing $P$ of $X$ is called {\bf connected}
if for any two elements $p_1, p_2 \in P$ there exists a connected  path in
$P$ from $p_1$ to $p_2$.
The property that there exists a connected  path
between two pairings is an equivalence relation on $P$, and we call
the equivalence classes {\bf connected   components} of $P$.
\end{definition}

  \begin{example} Let  $P$ be a pairing, whose pairs are  indicated by the black lines.
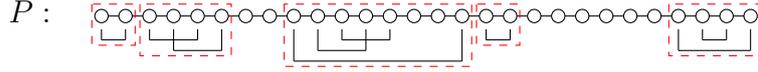
\begin{equation*}
P : \quad   \contraction{1-1-1-1-1-1-1-1-1-1-1-1-1-1-1-1-1-1-1-1-1-1-1-1-1-1-1-1}{[(1,2)],[(3,5),(4,6)_2],[(9,
    16)_3,(10,12)_2,(11,13)],[(17,18)],[(26,27),(25,28)_2]}
\end{equation*}
The connected components are indicated by the dashed boxes, that is
a connected component consist of all the pairs in a single dashed box.
\end{example}

We note that linked implies connected but not the other way around.

\begin{proof}[Proof of Proposition \ref{thm:algebraicenergies}]
  By Lemma \ref{lem:genwick} we can write
  \begin{align}
    T_n(r,\eta)
    &=\sum_{P \text{ pairing of } N_n   }   \sum_{\substack{\mathcal{I} \in  \mathcal{Q}(N_n) \\ N_n \notin \mathcal{I}    }}
    \mathcal{C}^0_P (\subst_{\substack{I \to -\EE_{|I|}\\ I \in \mathcal{I }} }(\pi_n(r,\eta))) \nonumber  \\
    &=  T_n^{(C)}(r,\eta) +  T_n^{(D)}(r,\eta)   , \label{eq:divideT}
  \end{align}
  where we divided the sum over the partitions into partitions which
  connect the smallest and the largest vertex
  \begin{align*}
    T_n^{(C)}(r,\eta) :=   \sum_{\substack{ P \text{ pairing of } N_n  \\ P \text{ links } 1 , n   }}   \sum_{\substack{\mathcal{I} \in  \mathcal{Q}(N_n)  \\ N_n \notin \mathcal{I}   }}
    \mathcal{C}^0_P  (\subst_{\substack{I \to -\EE_{|I|}\\ I \in \mathcal{I }} }(\pi_n(r,\eta)))
  \end{align*}
  (observe that in the above formula we can drop the condition
  $ N_n \notin \mathcal{I} $ because of $\{1,n \} \subset \bigcup P $) and
  the remaining partitions
  \begin{align}
    \label{eq:E_n:23}
    T_n^{(D)}(r,\eta) :=  \sum_{\substack{ P \text{ pairing of } N_n \\ P \text{ does not link } 1 , n    }}   \sum_{\substack{\mathcal{I} \in  \mathcal{Q}(N_n) \\ N_n \notin \mathcal{I}    }}
    \mathcal{C}^0_P  (\subst_{\substack{I \to -\EE_{|I|}\\ I \in \mathcal{I }} }(\pi_n(r,\eta))) .
  \end{align}
  To simplify the connected part $T_n^{(C)}(r,\eta)$ we proceed  as follows.
   We decompose the pairing  $P$ of $N_n$ which links $1$ and $n$, into linked components. We denote the linked component
  which contains  $\{1,n\}$ by $P_e$. Each of the  remaining linked components
  must be a pairing of $I$ for some  $I \in \mathcal{I}_{P_e}$,
  since otherwise the pairing  would link two elements of $N_n$, for which there
  would lie an element of $\bigcup P_e$ between them, a contradiction.
  Thus we  can write the pairing   $P$ of $N_n$ which links $1$ and $n$ in a unique way as
  \begin{equation} \label{eq:ppedcomp}
   P = P_e \cup \bigcup_{I \in \mathcal{I}_{P_e} } P_I ,
  \end{equation}
  where $P_I$ is a pairing of $I$.
  This  decomposition is  illustrated in the following example.

  \begin{example} Consider $N_{14}$ and let $P$ be the
   set of the pairs indicated by black lines.
  \begin{equation*}
  P : \quad    \contraction{1-1-1-1-1-1-1-1-1-1-1-1-1-1}{(1,5)_2,[(2,3)],(4,13)_3,[(6,
      8),(7,11)_2],(12,14)}
  \end{equation*}
  Then $P_e $ is the set of all pairs indicated by the lines
 which are outside of the dashed  boxes. Moreover,
   $\mathcal{I}_{P_e}= \{ I_1 , I_2 \}$ with $I_1 :=  \{2,3\}$, that
   is the set of  points in first  dashed box, and   $I_2 :=
   \{6,7,8,9,10,11\} $, that is the set of  points in the second dashed
   box.
  Furthermore $P_{I_1}$ is the set of all pairs indicated by the lines
  in the first dashed box, and similarly for $P_{I_2}$.
  \end{example}

Now using the decomposition \eqref{eq:ppedcomp} we obtain the first
identity of the following equations
\begin{align}
  &T_n^{(C)}(r,\eta)  \nonumber \\
  &\;=\sum_{\substack{  P_e  \text{ pairing of } N_n  \\ \{ 1,n  \}  \subset  \bigcup P_e    \\  P_e \text{ linked}   } }
  \prod_{I \in   \mathcal{I}_{P_e} } \bigg\{  \sum_{\substack{P_I
  \text{ pairing of     } I  } }  \bigg\} \nonumber \\
  & \qquad\qquad\qquad\qquad\qquad\;
  \sum_{\substack{\mathcal{I} \in  \mathcal{Q}(N_n)    }}
  \mathcal{C}^0_{(P_e \cup \bigcup_{I \in \mathcal{I}_{P_e}}  P_I) } (\subst_{\substack{J \to -\EE_{|J|}\\ J \in \mathcal{I}} }(\pi_n(r,\eta)))
  \nonumber \\
                                    &=
                                      \sum_{\substack{  P_e  \text{ pairing of } N_n  \\ \{ 1,n  \}  \subset  \bigcup P_e    \\  P_e \text{ linked}   }  }
  \prod_{I \in   \mathcal{I}_{P_e} } \bigg\{ \sum_{\substack{\mathcal{I}_I  \in  \mathcal{Q}(I)     }} \sum_{\substack{P_I  \text{ pairing of}   \\
   I \setminus \bigcup  \mathcal{I}_I} } \bigg\}
   \nonumber \\
  &\qquad\qquad\qquad\qquad\qquad\quad
  \mathcal{C}^0_{(P_e \cup \bigcup_{I \in \mathcal{I}_{P_e}}  P_I)  } (\subst_{\substack{J \to -\EE_{|J|}\\ J \in \bigcup_{I \in \mathcal{I}_{P_e}} \mathcal{I}_I} }(\pi_n(r,\eta)))
  \label{eq:E_n:233f} \\
                                    &=  \sum_{\substack{  P_e  \text{ pairing of } N_n  \\ \{ 1,n  \}  \subset  \bigcup P_e    \\  P_e \text{ linked }   } }
  \mathcal{C}^0_{P_e}( \subst_{\substack{I  \to \widehat{T}_{|I|}(\cdot + r,\eta)   \\ I \in \mathcal{I}_{P_e}}} (
  \pi_n(r,\eta)))  \label{eq:E_n:234g} \\
                                    &= C_n(r,\eta) ,  \nonumber
\end{align}
where in \eqref{eq:E_n:233f} we interchanged on each of the intervals
$I \in \mathcal{I}_{P_e}$ the summation on the one hand over energy
subtractions and on the other hand over pairings of $I$, and where
in \eqref{eq:E_n:234g} we used linearity of the contraction operator
$\CC_P$ and the multilinearity property of the product of graph
functions, see Lemma \ref{lem:algrelcont}.

To simplify  $T_n^{(D)}(r,\eta)$ we rearrange the summation by
decomposing it  into disconnected
parts.
To this end we define a bijection between  two
index sets.
The original index set is
 $$
 \mathcal{S}_1 := \{ (P , \mathcal{I}  ) : \mathcal{I} \subset \mathcal{Q}(N_n) , P \text{ is a pair partition of }
  N_n \setminus \bigcup \mathcal{I} \} .
 $$
Now suppose $(P, \mathcal{I} ) \in \mathcal{S}_1$ is given.
The following construction  is
illustrated in the example below.
First we consider the connected components of  $P$, and we  define the sets
\begin{align*}
\mathcal{K}_0  & := \{  S(Q) \cap \N  : Q \text{ connected component of } P \},   \\
\mathcal{I}_0  & := \{  K  \in \mathcal{I} : K \subset N_n \setminus  S(P)  \}, \\
\mathcal{K} & := \mathcal{K}_0 \cup \mathcal{I}_0 .
\end{align*}
Clearly, $\mathcal{K}$ is a partition of $N_n$ into intervals of $N_n$.
Next  we want to decompose $(P, \mathcal{I} )$  with respect to  the partition  $\mathcal{K}$.
For this  we  define for each $K \in \mathcal{K}$ the set
$$
\mathcal{I}_K := \{ I \subset  K : I \in \mathcal{I} \} ,
$$
moreover,   we define for $K \in \mathcal{K}_0$  the pairing
\begin{align*}
P_{K}  := Q ,
\end{align*}
where $Q$ is the unique connected component of $P$ such that $K =  S(Q) \cap \N $,
and for $K \in \mathcal{I}_0  $  we define
$
P_K = \emptyset .
$
It  is straight forward to verify that this construction yields a well defined  map $\psi$ from the set $\mathcal{S}_1$
to the index set
\begin{align*}
 \mathcal{S}_2 &  := \{ (\mathcal{K}, (P_K)_{K \in \mathcal{K}}, (\mathcal{I}_K)_{K \in \mathcal{K} }   ) : \mathcal{K} \in  \mathcal{Q}(N_n) ,
 \bigcup \mathcal{K} = N_n , \\
 & \qquad \quad  \mathcal{I}_K \in \mathcal{Q}(K) , P_K \text{ is a pair partition of } K \setminus \bigcup \mathcal{I}_K , \\
& \qquad  \quad  ( S(P_K) = [\min K , \max K ] \text{ or } (  P_K = \emptyset \text{ and } \mathcal{I_K} = K   ))
  \}  .
 \end{align*}
 In fact $\psi$ is a bijection with inverse
 $$
 \psi^{-1}(\mathcal{K}, (P_K)_{K \in \mathcal{K}}, (\mathcal{I}_K)_{K \in \mathcal{K} }   ) =  ( \bigcup_{K \in \mathcal{K}} P_K , \bigcup_{K \in \mathcal{K}}
 \mathcal{I}_K ) ,
  $$
  as one readily verifies.

  \begin{example} We consider a  pairing $P$ of $N_{28}$, whose pairs are  indicated by the black lines below. The
  set    $\mathcal{I} = \{ I_1, I_2 , I_3 , I_4 \}$ with  $I_1 :=  \{7,8\}, I_2 :=  \{14,15\} , I_3 :=  \{19,20\}, I_4 := \{ 21,22,23,24 \}$
  is indicated below as well. Likewise  the sets   $\mathcal{K}_0$, $\mathcal{I}_0$,  and $\mathcal{K}$  are indicated.
  
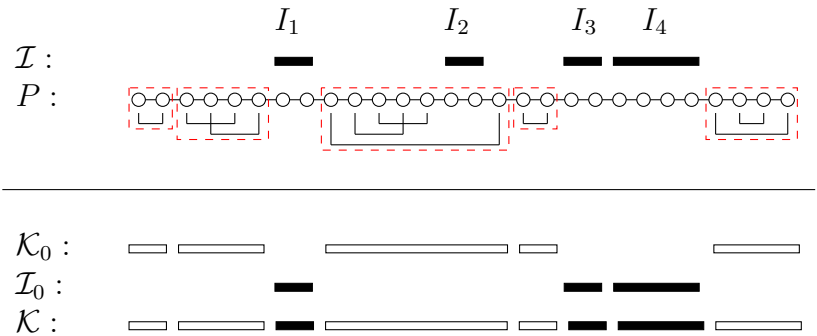
\begin{equation*}
\begin{array}{ll}
   &
    \hspace{1.92cm}     I_1        \hspace{1.74cm}  \  I_2  \  \hspace{1.05cm} \  I_3    \hspace{0.17cm}  \   \ \  I_4  \     \\
 \mathcal{I} : \quad   &  \hspace{1.92cm}    \contraction{2srf}{}     \hspace{1.74cm}  \contraction{2srf}{} \hspace{1.05cm} \contraction{2srf}{} \hspace{0.145cm}
 \contraction{4srf}{}
        \\
 P : \quad &    \contraction{1-1-1-1-1-1-1-1-1-1-1-1-1-1-1-1-1-1-1-1-1-1-1-1-1-1-1-1}{[(1,2)],[(3,5),(4,6)_2],[(9,
    16)_3,(10,12)_2,(11,13)],[(17,18)],[(26,27),(25,28)_2]}
     \\
    \ & \  \\
     \hline \\
    \mathcal{K}_0 : \quad &
    \contraction{2sr}{} \hspace{0.145cm} \contraction{4sr}{}  \hspace{0.8cm}
    \contraction{8sr}{} \hspace{0.145cm} \contraction{2sr}{} \hspace{2.05cm}  \contraction{4sr}{}
     \\
    \mathcal{I}_0 : \quad &
    \hspace{1.92cm} \contraction{2srf}{}  \hspace{3.3cm}
   \contraction{2srf}{} \hspace{0.145cm} \contraction{4srf}{}
     \\
    \mathcal{K} : \quad &
    \contraction{2sr}{} \hspace{0.145cm} \contraction{4sr}{}  \hspace{0.145cm} \contraction{2srf}{}  \hspace{0.145cm}
    \contraction{8sr}{} \hspace{0.145cm} \contraction{2sr}{} \hspace{0.145cm} \contraction{2srf}{} \hspace{0.145cm} \contraction{4srf}{} \hspace{0.145cm} \contraction{4sr}{}
     \end{array}
\end{equation*}
\end{example}

Furthermore,  one sees that
\begin{align*}
\psi ( \{ (P,\mathcal{I}) \in \,&\mathcal{S}_1 :  S(P) \neq [1,n] , \,   N_n \notin \mathcal{I}  \} )   \\ &=
\{  (\mathcal{K}, (P_K)_{K \in \mathcal{K}}, (\mathcal{I}_K)_{K \in \mathcal{K} } ) \in \mathcal{S}_2   :  |\mathcal{K} | \geq 2 \}  .
\end{align*}
  Thus the bijection  allows us to rearrange the
sum  in \eqref{eq:E_n:23}  and we obtain the first identity of the following equations
\begin{align}
  &T_n^{(D)}(r,\eta) \nonumber  \\
  &\quad=  \sum_{\substack{ \mathcal{K} \in \mathcal{Q}(N_n) \\  \bigcup
  \mathcal{K} =  N_n \\ |\mathcal{K}| \geq 2}} \prod_{K \in
  \mathcal{K}}
  \bigg\{ \sum_{\substack{ P_K \text{ pairing of } K \\  S(P_K) = \\  [\min  K , \max K]
  }} \sum_{\substack{\mathcal{I}_K  \in  \mathcal{Q}(K)     }}
  +    \ 1_{(P_K = \emptyset , \mathcal{I}_K = \{ K \}  ) }  \bigg\} \nonumber \\
 & \qquad\qquad\qquad\qquad\qquad\qquad\quad \times
 \mathcal{C}^0_{(\bigcup_{K \in \mathcal{K}}  P_K)} (
  \subst_{\substack{I \to -\EE_{|I|}\\ I \in  \bigcup_{K \in
  \mathcal{K} } \mathcal{I}_K }}(\pi_{n}(r,\eta)))) \nonumber \\
  &\quad=  \sum_{\substack{ \mathcal{K} \in \mathcal{Q}(N_n) \\  \bigcup
  \mathcal{K} =  N_n \\ |\mathcal{K}| \geq 2}} \prod_{K \in
  \mathcal{K}}
  \bigg\{ \sum_{\substack{ P_K \text{ pairing of } K  \\  S(P_K) = \\  [\min  K , \max K]
  }} \sum_{\substack{\mathcal{I}_K  \in  \mathcal{Q}(K)     }}
  +    \ 1_{(P_K = \emptyset , \mathcal{I}_K = \{ K \}  ) } \bigg\}  \nonumber  \\
  &\qquad\qquad\qquad \times \prod_{K  \in \mathcal{K} \setminus \max \mathcal{K} }
   \bigg\{  \mathcal{C}^0_{P_K}  (\subst_{\substack{I \to -\EE_{|I|}\\ I
  \in \mathcal{I}_K }}(\pi_{|K|}(r,\eta)))  R(r,\eta) \bigg\} \nonumber \\
    & \qquad\qquad\qquad\qquad\qquad\qquad\qquad \times
    \mathcal{C}^0_{P_{\max \mathcal{K}} }
  (\subst_{\substack{I \to -\EE_{|I|}\\ I \in \mathcal{I}_{\max
  \mathcal{K}} }}(\pi_{|\max \mathcal{K}|}(r,\eta)))  \label{eq:E_n:233db} \\
  &\quad= \sum_{k=2}^n \sum_{\substack{j_1+\ldots + j_k = n\\ j_i \geq
  1}} \left[ \prod_{i=1}^{k-1}\bigl( \widehat{C}_{j_i}(r,\eta)
  R(r,\eta) \bigr)
  \right] \widehat{C}_{j_k}(r,\eta), \label{eq:E_n:232dff}
\end{align}
where in \eqref{eq:E_n:233db} we made use of the fact that the
contraction factors for disconnected parts. Moreover  the product over
$\mathcal{K}$ is taken with respect to the following ordering.  An
element $ \{I_1,\dotsc,I_l\}$ of $\mathcal{Q}(M)$ has an ordering
given by
\begin{equation} \label{eq:orderingintervals} I_i \leq I_j \quad :
  \Leftrightarrow \quad a \leq b , \forall a \in I_i , \forall b \in
  I_j ,
\end{equation}
which is total and well ordered.  Finally we note that Eq.
\eqref{eq:E_n:232dff} follows from multilinearity and the
definitions, where the renormalization terms $\mathcal{E}_{j_i}$ in $\widehat{C}_{j_i}$, see \eqref{defofchat},  come from the terms $1_{(P_K = \emptyset , \mathcal{I}_K = \{ K \}  ) }$.
\end{proof}

In order to estimate the Feynman graphs, we will decompose the
resolvent in  Proposition \ref{thm:algebraicenergies}.
We write
\begin{equation}  \label{eq:decofresolv}
R(r,\eta) = R^\perp(r,\eta) + R^\parallel(r,\eta) ,
\end{equation}
where we have defined
$$
R^\perp(r,\eta) = (1-P_{\rm at}) R(r,\eta) , \quad R^\parallel(r,\eta)
= P_{\rm at} R(r,\eta) .
$$
We define
\begin{equation} \label{eq:defofg}
G_n(r,\eta) := C_n(r,\eta) + \sum_{k=2}^n \sum_{\substack{j_1+\ldots +
    j_k = n\\ j_i \geq 1}} \left[ \prod_{i=1}^{k-1}\bigl(
  \widehat{C}_{j_i}(r,\eta) R^\perp(r,\eta) \bigr) \right]
\widehat{C}_{j_k}(r,\eta)
\end{equation}
and
$$
\widehat{G}_n(r,\eta) := G_n(r,\eta) - \EE_n .
$$

The next theorem is purely algebraic. It will later be used to
estimate the energy coefficients.

\begin{proposition}\label{thm:algebraicenergy}
  We have
  \begin{align}
    T_n(r,\eta) & = G_n(r,\eta)
                +  \sum_{k=2}^n \sum_{\substack{j_1+\ldots + j_k = n\\ j_i \geq
    1}} \left[ \prod_{i=1}^{k-1}\bigl( \widehat{G}_{j_i}(r,\eta) R^\parallel(r,\eta) \bigr)
    \right] \widehat{G}_{j_k}(r,\eta)   \label{eq:algebraicenergy}
  \end{align}
  and
\begin{equation} \label{eq:algebraicenergyhat}
\widehat{T}_n(r,\eta) = \widehat{G}_n(r,\eta) + \sum_{k=2}^n
\sum_{\substack{j_1+\ldots + j_k = n\\ j_i \geq 1}} \left[
  \prod_{i=1}^{k-1}\bigl( \widehat{G}_{j_1}(r,\eta)
  R^\parallel(r,\eta) \bigr) \right] \widehat{G}_{j_k}(r,\eta) .
\end{equation}
\end{proposition}
\enlargethispage{2cm}
\begin{proof} In view of the previous proposition it remains to
  decompose the resolvent between disconnected parts into orthogonal
  and parallel part.  To this end we multiply out the expressions and
  collect the terms according to the number, $s-1$, of resolvents with
  $R^\parallel$.
  Thus by straight forward algebraic calculation we obtain
  \begin{align*}
  &\sum_{k=2}^n \sum_{\substack{j_1+\ldots + j_k = n\\ j_i \geq
    1}} \left[ \prod_{i=1}^{k-1}\bigl( \widehat{C}_{j_i}(r,\eta) R(r,\eta) \bigr)
    \right] \widehat{C}_{j_k}(r,\eta)  \nonumber \\
 &\quad=
   \sum_{k=2}^n \sum_{\sigma_1,\dotsc,\sigma_{k-1} \in \{ \perp, \parallel \}}
   \sum_{\substack{j_1+\ldots + j_k = n\\ j_i \geq
    1}} \left[ \prod_{i=1}^{k-1}\bigl( \widehat{C}_{j_i}(r,\eta) R^{\sigma_i}(r,\eta) \bigr)
    \right] \widehat{C}_{j_k}(r,\eta) \nonumber  \\
 &\quad=  \sum_{s=1}^n   \sum_{\substack{n_1+\ldots + n_s = n \\ n_i \geq   1}}
    \sum_{\substack{ k_1,\dotsc, \, k_s \in N_n  \\ k_1 + \cdots + k_s \geq 2 } }  \sum_{\substack{j_{1,1} +\ldots + j_{1,k_1}  = n_1 \\ j_{1,i}  \geq  1}} \cdots   \sum_{\substack{j_{s,1} +\ldots + j_{s,k_s}  = n_s \\ j_{s,i}  \geq  1}} \nonumber
    \\
 &\qquad\qquad\qquad\quad  \left[ \prod_{i_1=1}^{k_1-1}\bigl( \widehat{C}_{j_{1,i_1}}(r,\eta) R^{\perp}(r,\eta) \bigr)
    \right] \widehat{C}_{j_{1,k_1}}(r,\eta) R^{\parallel}(r,\eta)  \cdots
    \nonumber \\
    & \qquad\qquad\qquad\qquad\qquad\qquad\quad
    \left[ \prod_{i_s=1}^{k_s-1}\bigl( \widehat{C}_{j_{s,i_s}}(r,\eta) R^{\perp}(r,\eta) \bigr)
   \right] \widehat{C}_{j_{s,k_s}}(r,\eta) \nonumber \\
 &\quad=
   \sum_{k=2}^n \sum_{\substack{j_1+\ldots + j_k = n\\ j_i \geq
    1}} \left[ \prod_{i=1}^{k-1}\bigl( \widehat{C}_{j_i}(r,\eta) R^\perp(r,\eta) \bigr)
    \right] \widehat{C}_{j_k}(r,\eta) \nonumber \\
 & \qquad\qquad  +  \sum_{s=2}^n \sum_{\substack{n_1+\ldots + n_s = n\\ n_i \geq
    1}} \left[ \prod_{i=1}^{s-1}\bigl( \widehat{G}_{n_i}(r,\eta) R^\parallel(r,\eta) \bigr)
    \right] \widehat{G}_{n_{s}}(r,\eta) , \nonumber
  \end{align*}
  where we use  the convention that an empty product is defined as   a multiplicative identity.
  In particular, in the fourth and fifth line the  empty product $\prod_{i=1}^{k_\nu-1} \cdots$ with  $k_\nu = 1$ is interpreted as a one.
  Moreover, on the last right hand side the first term originates from
  $s=1$ and the second term from summing over all $s \geq
  2$. Collecting equalities yields the claim.
\end{proof}

\subsection{Estimating the  Renormalized Graphs}
\label{subsec:estimateandProofThm1}

In this subsection we will prove  Theorem~\ref{thm:mainenergy}.  First we recall Lemma~\ref{lem:energyform2}
which relates  $T_n$ to the  expansion coefficients of the ground state energy.
To prove Theorem~\ref{thm:mainenergy} we  use   the formula for $T_n$ given in \eqref{eq:algebraicenergy}.
In the following lemmas below  we give   a few abstract inequalities  which will be needed
to estimate the expression in \eqref{eq:algebraicenergy}.
To show that  \eqref{eq:algebraicenergy} is indeed finite for $r \to 0$ and $\eta \to 0$ we use an
induction argument, which is sketched in the following remark.

\begin{remark}\label{rem:induction} To show Theorem~\ref{thm:mainenergy} we shall make the
  induction hypothesis  that $C_m$ and $G_m$, defined in \eqref{eq:defofc} and \eqref{eq:defofg},
  are sufficiently regular and $P_{\rm at} \widehat{G}_m(0,0) P_{\rm at} = 0$ for all $m \leq n$.
Then it will follow from   \eqref{eq:algebraicenergy} and  \eqref{eq:algebraicenergyhat}  that also $T_m$ is sufficiently regular and $P_{\rm at} \widehat{T}_m(0,0) P_{\rm at} = 0$ for all $m \leq n$.  The singularity of the resolvent at $r=0$  is cancelled,  since by induction hypothesis $G_m$ is sufficiently regular
and $P_{\rm at} \widehat{G}_m(0,0) P_{\rm at} = 0$ for all $m \leq n$.
Using the estimates of the lemmas below we will then see that $C_{n+1}$ is sufficiently regular, and thus also $G_{n+1}$ in view of  \eqref{eq:defofg}.
Now  from   \eqref{eq:algebraicenergy} and \eqref{eq:algebraicenergyhat} it will follow  that $P_{\rm at} \widehat{G}_{n+1}(0,0) P_{\rm at} = 0$,
where the singularity  is cancelled as before.  Hence the induction hypothesis for $n+1$  holds.
\end{remark}

\begin{lemma} \label{lem:remlast0}\label{lem:remlast} Let $X$ be a finite subset of $\Z$ containing at least four elements.
Let $P$ be a linked pair partition of    $X$,   and let $p \in P$. Then there exists a pair $q \in P$
  different from $p$ such that $P \setminus \{ q \}$ is again a linked
   pair partition of $X \setminus q$.
\end{lemma}

\begin{example}\label{ex:removal} Consider the set $X = \{ x_1, x_2 , \dotsc , x_{10} \} \subset \Z$, with
$$
x_1 < x_2 < \cdots < x_{10} ,
$$
which is indicated by the circles in the diagram below (where the index increases from left to right).
We consider the pair partition $P$ which is indicated by the lines below.
 
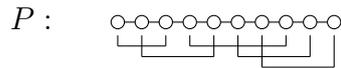
\begin{equation*}
\begin{array}{ll}
 P : \quad &    \contraction{1-1-1-1-1-1-1-1-1-1}{(1,3)_1,(2,5)_2,(4,8)_1,(6,9)_2,(7,10)_3}
   \end{array}
\end{equation*}
If we remove one of the pairs $\{x_2, x_5 \}$ or  $\{x_4, x_8 \}$,
then the set of the remaining  pairs is not linked anymore.
If we remove one of the three pairs $\{x_1, x_3 \}$, $\{x_6, x_9 \}$, or  $\{x_7,x_{10}\}$,  then  the
set of the remaining pairs is linked. Since we can remove any of the aforementioned three pairs, it
  follows that for any $p \in P$ there  exists a $q \in P$
  different from $p$ such that $P \setminus \{ q \}$ is again a linked
   pair partition of $X \setminus q$.
\end{example}

\begin{proof}
  For any two $r, s \in P$ define the distance
$$
d_P(r,s) :=  \inf \{ l \in \N : \text{there exists a linked path in}\, P\,
\text{of length}\, l\, \text{from}\,  r \,\text{to}\, s \}
$$
Clearly,  $d_P$ is a metric on $P$.  Define
$$
m_P(p) := \max\{d(p,r) : r \in P \} .
$$
Since $X$ is finite we can pick  a $q \in P$ such that
$$
d_P(p,q) = m_P(p).
$$
Then $P \setminus \{ q \}$ is again linked. Otherwise there would
exist at least two linked components. One component must contain $p$ and we could pick an  $r \in P \setminus \{ q \}$ in a different component. But then   every linked path in
$P$ from $p$ to $r$ would have to  pass through $q$. This would imply
$d_P(p,q) < d_P(p,r) \leq m_P(p)$. This is a contradiction to the choice
of $q$.
\end{proof}

\begin{lemma}
  \label{lem:lemestimate}
  Let $X \subset \Z$ be  a finite set. Let $P$ be a  linked pair partition of $X$.
  \begin{enumerate}[label=\rm (\alph*)]
\item \label{lem:lemestimate:a}
   Suppose we are given,
 $$
 \phi = ( (V)_{x \in X} , (F_{e})_{e \in {E}_X} ) ,
 $$
 a  $(\mathcal{V}_{\rm sb}, \mathcal{R}_{\rm sb})$-valued  graph function  on $X$.
 Suppose there exists a   constant $C_F$  such
  that for all  $e \in {E}_X$ we have
   $$\|F_e(r) \| \leq C_F (|r|^{-1} + 1 ) , \forall r \geq    0  . $$
  Then
$$
\sup_{r \geq 0} \| \CC_P(\phi)(r) \| \leq  C_F^{|X|} C_1^{|X|-2} C_0^2 ,
$$
where
$$
C_p := \left( \int dk (|k|^{-1} +1 )^{p+1} \| G(k) \|^2 \right)^{1/2} .
$$
\item \label{lem:lemestimate:b}
  Let $S \subset \R^d$.
  Suppose for each $s \in S$ we  are given
 $$
 \phi_s = ( (V)_{x \in X} , (F_{e,s})_{e \in {E}_X} ) ,
 $$
  a  $(\mathcal{V}_{\rm sb}, \mathcal{R}_{\rm sb})$-valued  graph function  on $X$.
  Suppose for each $r > 0$ and for each $e \in  {E}_X$, the function $s \mapsto F_{e,s}(r)$ is continuous,
  and suppose there exists a   constant $C_F$  such
  that
   $$\|F_{e,s}(r) \| \leq C_F (|r|^{-1} + 1 ) , \quad \forall  r  >    0 , s \in S , e \in {E}_X  . $$
  Then  the function $s \mapsto  \CC_P(\phi_s)(r)$ is continuous for each $r \geq 0$.
\end{enumerate}
\end{lemma}

\begin{proof}
\ref{lem:lemestimate:a}
  Since $P$ is a pair partition of $X$ the cardinality of $X$ must be even.
  Thus we have $|X|=2n$ for some $n \in \N$.  Using the  notation
  introduced  in Definition \ref{def:defofcont} we estimate
  \begin{align*}
       \| \CC_P( \phi )(r) \| &\leq  \int \prod_{x \in X } \{ dk_x \}
      \delta_{P}(k) \nonumber \\
      & \quad\;  \prod_{j \in X \setminus \{ \max X \}} \big\{ \| G^\sharp_{j,P}(k_j) \| \|
      F_{\{j,r_X(j)\}}(r + |K_{\{j,r_X(j)\}}|_P) \|
      \big\} \nonumber \\
      &\qquad\qquad\qquad\qquad \times \| G^\sharp_{ \max X, P}(k_{\max X})  \| \nonumber  \\
      &\leq  {\rm Est}_n,
  \end{align*}
  where we define
  \begin{align}     \label{eq:linkedgraphest2d}
   {\rm Est}_n := \sup_{\substack{P \text{ linked pair} \\
       \text{ partition of } X } }
     & \int  \prod_{x \in  X}  \{ dk_x \}  \delta_{P}(k)  \nonumber  \\
      & \; \prod_{j \in X \setminus \{ \max X \} } \left\{ \| G^\sharp_{j,P}(k_j) \|  C_F
        (|K_{\{j,r_X(j)\}}|_P^{-1} + 1 ) \right\} \nonumber \\
       &\qquad\qquad\qquad\qquad \times  \| G^\sharp_{\max X, P}(k_{\max X})  \|.
  \end{align}
   We will show by induction
  in $n$ that
$$
{\rm Est}_n \leq C_F^{2n-1}  C_1^{2n-2}C_0^2 .
$$
First we consider the case $n=1$.
\begin{align*}
  {\rm Est}_1 &=
     \int  dk_{\min X}  dk_{\max X} \Big( \delta( k_{\min X} - k_{\max X} )
     \\ &\qquad\qquad\quad \| {G}^{*}(k_{\min X})  \|
          \ C_F (|k_{\max X}|^{-1} + 1)   \| G(k_{\max X})  \| \Big)\\
     &\leq  C_F C_0^2.
\end{align*}
Next we show the induction step $n-1 \to n$. The goal is to integrate
out a pair of paired variables, such that the set of pairings of the remaining
variables  remains  linked.  The details, which are
illustrated in the  example below, are as follows.
Let $P$ be a linked pair partition of $X$, such that the supremum in
 \eqref{eq:linkedgraphest2d} is attained at $P$.
By Lemma \ref{lem:remlast0}, we can pick a pair
$q \in P$ such that $P_q := P \setminus \{ q \}$ is a
linked pair partition of $X_q := X \setminus q$.
We want to remove  the propagators over the edges, which are adjacent to $q$ and lie in the span of $q$,
i.e.  the edges
$$
e_{\rm l}  := {\{\min q, r_X({\min q}) \}, \quad e_{\rm r}  := \{l_X({\max q}),\max q\}} ,
$$
where we introduced the notation  $l_X(x) := \max( X \setminus [x,\infty))$ (denoting the
nearest neighbor on the left  of $x$ in $X$).
To this end we will use the  following lower bounds, which are
an immediate consequence of the definition,
\begin{equation} \label{eq:estonvar1}
|K_{e_{\rm l}   } |_P
  \geq | k_{\max q} | , \quad |K_{e_{\rm r}   } |_P
  \geq | k_{\max q} | ,
\end{equation}
\begin{equation} \label{eq:estonvar2}
 |K_{e}|_{P} \geq  |K_{e}|_{P_q}   , \quad \forall e \in E_X  \cap E_{X_q} ,
\end{equation}
and
\begin{align}
|K_{\{ l_X({\min q}), \min q  \}} |_P &
  \geq |K_{\{  l_X({\min q}),r_X({\min q}) \}} |_{P_q}  , \ \ \text{ provided } \min q \neq \min X   ,
  \label{eq:estonvar3} \\
  |K_{\{q, r_X({\max q})  \} } |_P &
  \geq |K_{\{ l_X({\max q}), r_X({\max q}) \}} |_{P_q}  ,  \ \ \text{ provided } \max q \neq \max X
  \label{eq:estonvar4} .
\end{align}
We use  the   estimates in  \eqref{eq:estonvar1}--\eqref{eq:estonvar4} to obtain
an upper bound for \eqref{eq:linkedgraphest2d}, integrate
 out the variables $k_{\max q}$ and $k_{\min q}$, which are paired by a delta function,
and  use  the inequality
\begin{align*}
  \int d & k_{\max q} dk_{\min q} \Big( \delta( k_{\max q} - k_{\min q} )
	\\
   &\quad  ( | k_{\max q}  |^{-1} + 1 ) (1 +  |  k_{\max q} |^{-1} )
  \| G(k_{\min q})\| \|G(k_{\max q})\| \Big)
  \leq C_1^2 .
\end{align*}
This yields the following estimate
\begin{align*}
  {\rm Est}_n \leq {}&
  C_F^2  C_1^2 \int \prod_{x \in  X_q}  \{ dk_x   \} \delta_{P_q}(k)  \\
              &\qquad\quad \prod_{j \in X_q  \setminus \{ \max X_q \} }
              \left\{ \|  G^\sharp_{j,P_q}(k_j) \|
                 \ C_F (|K_{\{j, r_{X_q}(j) \}}|_{P_q}^{-1} + 1 )   \right\}
                 \\&\qquad\qquad\qquad\qquad\qquad\; \times
                 \| G^\sharp_{\max X_q, P_q}(k_{\max X_q})  \| \\
               \leq {}& C_1^2 C_F^2  {\rm Est}_{n-1} .
\end{align*}
This shows  the induction step.

\begin{example} Consider the situation as in  Example \ref{ex:removal} above.
If we choose $q = \{ x_6, x_9 \}$, then $P \setminus \{q \}$ is a linked pair partition of $X \setminus \{ q \}$.
In that case we want to remove the propagators over the edges
 $$
 e_{\rm l}   =  \{x_6,x_7\} , \quad e_{\rm r}  =  \{x_8,x_9\}
 $$
 and we integrate over the pair of
variables $k_{{\rm min} q} = k_{x_6}$ and $k_{{\rm max} q} = k_{x_9}$.
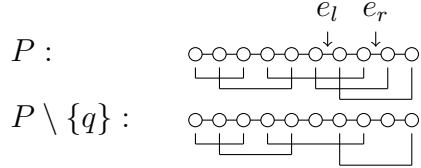
 \begin{equation*}
\begin{array}{ll}
 P : \quad &    \contraction{1-1-1-1-1-1-<$e_l$>1-1-<$e_r$>1-1}{(1,3)_1,(2,5)_2,(4,8)_1,(6,9)_2,(7,10)_3} \\
 P \setminus \{ q  \} : \quad &    \contraction{1-1-1-1-1-1-1-1-1-1}{(1,3)_1,(2,5)_2,(4,8)_1,(7,10)_3}
   \end{array}
\end{equation*}
\end{example}

\ref{lem:lemestimate:b}
 This follows from dominated convergence and part \ref{lem:lemestimate:a}. Explicitly
we estimate  for $s,t \in S$ and all $r \geq 0$

 \begin{equation*}
    \begin{split}
      \| \CC_P(& \phi_s )(r) -  \CC_P( \phi_t )(r) \| \\
    & \;  \leq {} \!\!\! \sum_{l \in X \setminus \{ \max X \}  } \int \prod_{x \in X } \{ dk_x \}
      \delta_{P}(k)
      \\
      & \qquad \prod_{j \in X \setminus \{ \max X \}}  \Big\{ \| G^\sharp_{j,P}(k_j) \| \|
      1_{j \leq l}  F_{\{j,r_X(j)\},s}(r + |K_{\{j,r_X(j)\}}|_P)  \\
      &\qquad\qquad\qquad\qquad\qquad\qquad\quad\; -    1_{j \geq l}   F_{\{j,r_X(j)\},t}(r + |K_{\{j,r_X(j)\}}|_P) \| \Big\} \\
      & \qquad\qquad\quad \times \| G^\sharp_{ \max X, P}(k_{\max X})  \| .
    \end{split}
  \end{equation*}
The factor  for $l=j$ converges to zero. The integrand can be estimated by using the triangle inequality for the factor $l=j$.
This results in an integrand which is an upper bound and which can be integrated by the proof of \ref{lem:lemestimate:b}.
\end{proof}

We need an analogous estimate on the derivative.

\begin{lemma}\label{lem:lemestimate2}
  Let $X \subset \Z$ be  a finite set and   let $S \subset \R^d$ be open.
  Suppose for each $s \in S$ we  are given
 $$
 \phi_s = ( (V)_{x \in X} , F_{e,s})_{e \in {E}_X} ) ,
 $$
  a  $(\mathcal{V}_{\rm sb}, \mathcal{R}_{\rm sb})$-valued  graph function  on $X$.
  Suppose there exist  constants $C_F$  such
  that for all  $e \in {E}_X$ we have
   $$\|F_{e,s}(r) \| \leq C_F (|r|^{-1} + 1 ) , \quad  r  >    0  . $$
 Suppose   that for each $r > 0$ and for each $e  \in {E}_X$, the function $s \mapsto F_{e,s}(r)$ has
 continuous partial
   derivatives with
  $$
  \|\partial_{s_i} F_{e,s}(r) \| \leq C_F (|r|^{-1} + 1 )^2,  \quad   r > 0 , \quad i = 1, \ldots, d.
  $$
  Then for any linked pair partition $P$ of $X$ the function $s \mapsto  \CC_P(\phi_s)(r)$ has  for each $r \geq  0$
  continuous partial derivatives and
$$
\| \partial_{s_i}  \CC_P(\phi_s(r)) \| \leq (|X|+1)  C_F^{|X|-1}
C_1^{|X|} .
$$
\end{lemma}

\begin{proof}
  As in the proof of the previous lemma  we can assume that $|X|=2n$ for some
  $n \in \N$.  Using the well known arguments  ensuring the interchange of
  differentiation and integration, we essentially need to show that
  the norm of the differentiated integrand can be estimated from above by an integrable function.
  (Strictly speaking, only the estimates given in the following proof
  will justify  the existence
  of the partial derivative $ \partial_{s_i} \CC_P( \phi_s(r) )$. Nevertheless for notational
  compactness we shall
  already write $ \partial_{s_i} \CC_P( \phi_s(r) )$  in \eqref{eq:estonfirstder}  and \eqref{eq:estonfirstder2}, below.)
  We calculate the derivative using Leibniz's rule.
    First we show the estimate in case $n=1$, in which case  we find
  \begin{equation}
    \label{eq:estonfirstder}
    \begin{split}
      \| \partial_{s_i} \CC_P( \phi_s(r) ) \|
      \leq &{}  \int dk_{\min X}
      dk_{\max X} \delta( k_{\min X} - k_{\max X} ) \| {G}^{*}(k_{\min
        X}) \|
      \\
      &\times C_F ((|k_{\max X}|+r)^{-1} + 1 )^2   \| G(k_{\max X})  \|\\
      \leq &{}  C_F C_1^2 .
    \end{split}
  \end{equation}
  Let us now consider the case $n \geq 2$.  Calculating the derivative
  using Leibniz's  rule we find that for all $r > 0$ and $n \geq 1$ we
  have
  \begin{align}    \label{eq:estonfirstder2}
    \|  \partial_{s_i} \CC_P(  \phi_s(r)  ) \|
    \leq {}&  \sum_{l \in X \setminus \{ \max X \} } \int \prod_{x \in  X}  \{ dk_x   \}
             \delta_{P}(k)  \\
           & \quad\;
             \prod_{j \in X  \setminus \{ \max X \} }  \left\{ \|     G^\sharp_{j,P}(k_j) \|
              C_F (|K_{\{j,r_X(j)\}}|_{P}^{-1} + 1 )^{1 +
             \delta_{jl}}    \right\} \nonumber \\
            & \qquad\qquad\qquad\qquad \times \| G^\sharp_{\max X, P}(k_{\max X})  \| \nonumber
    \\
    \leq {}&  ( |X|-1){\rm DEst}_{n}   , \nonumber
  \end{align}
  where we defined
  \begin{align}
    \label{eq:linkedgraphest2dd}
    {\rm DEst}_{n}
    := & \sup_{l \in  X \setminus \{ \max X \}  } \sup_{\substack{P   \text{ linked}  \\  \text{pair partition of } X   }  }
    \int \prod_{x \in  X}  \{ dk_x   \}
    \delta_{P}(k) \\
    &\qquad
    \prod_{j \in X  \setminus \{ \max X \} }  \left\{ \|    G^\sharp_{j,P}(k_j) \|
      \ C_F (|K_{\{j,r_X(j) \}}|_{P}^{-1} + 1 )^{1 + \delta_{jl}}   \right\}
      \nonumber \\
      & \qquad\qquad\qquad\qquad\quad \times
      \| G^\sharp_{\max X,P}(k_{\max X})  \|  . \nonumber
  \end{align}
  We want to show by induction that for all $n \in \N$ we have
      $$
      {\rm DEst}_{n} \leq C_F^{2n-1} C_1^{2n} .
      $$
      By \eqref{eq:estonfirstder} we know that  the inequality for ${\rm DEst}_{1}$
      holds.
      Next we show that $n - 1 \to n$. Let us first sketch the idea.
        As in the proof of Lemma \ref{lem:lemestimate},
       we want to remove two propagators by integrating    out a
      pair of paired variables, such that the set of pairings of the remaining
      variables is again linked, in addition, we want the term which contains
      a two in the exponent to remain in the integral.
      The details are as follows, and illustrated in the examples below.
      Suppose the maximum  in  \eqref{eq:linkedgraphest2dd} is attained for
        $l \in  X \setminus \{ \max X \} $ and the linked pairing $P$.
      Consider the  edge $e := \{l,r_X(l)\}$ (for which we have a two in the exponent).
       We pick a  $p \in P$ such that $e$
       lies in the span of $p$ (this can always be achieved, since  $P$ is linked
       and therefore connected).  By Lemma \ref{lem:remlast} there exists an element
      $q \in P \setminus \{p \}$ such that $P_q := P \setminus \{ q \}$
      is a linked  pair partition of $X_q = X \setminus q $.
      If none of the edges  $$
e_{\rm l}  := {\{\min q, r_X({\min q}) \}, \quad e_{\rm r}  := \{l_X({\max q}),\max q\}} ,
$$
is equal to $e$, then an  estimate as in the proof of Lemma \ref{lem:lemestimate}
yields
  \begin{align}\label{eq:estondestn}
        {\rm DEst}_n     &\leq   C_1^2 C_F^2  {\rm DEst}_{n-1} ,
  \end{align}
      since the term involving $e$ is not integrated out.
      If $e = e_{\rm l} $ (the case $e = e_{\rm r} $ is analogous) then we use the estimate
      $$
      (|K_{e}|_{P}^{-1} + 1 )^2 \leq   (|K_{\{l_X({\min q}), r_X({\min q }) \}} |_{P_q}^{-1} + 1 )  (|k_{\max q}|^{-1} + 1 )  .
      $$
      The second term on the right hand side is again estimated as in the proof of Lemma \ref{lem:lemestimate}, whereas
      the first term remains in the integral. This yields again \eqref{eq:estondestn}.
      The continuity of the derivative follows from dominated convergence as
      in the proof of
      Lemma~\ref{lem:lemestimate}.

\begin{example} Consider the situation as in  Example \ref{ex:removal} above.
Suppose $l = x_6$ and so $e = \{x_6,x_7\}$. Then for $p = \{x_6,x_9\}
\in P$ the span of $p$ contains  $e$.
If we choose  $q = \{x_1,x_3\}$, then $P \setminus \{q \}$
is a linked pair partition of $X \setminus \{ q \}$. In that case
 we want to remove the propagators over the edges
 $$
 e_{\rm l}   =  \{x_1,x_2\} , \quad e_{\rm r}  =  \{x_2,x_3\}  ,
 $$
 which are both different from $e$.
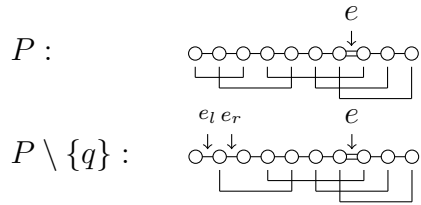
 \begin{equation*}
\begin{array}{ll}
 P : \quad &    \contraction{1-1-1-1-1-1-1=<$e$>1-1-1}{(1,3)_1,(2,5)_2,(4,8)_1,(6,9)_2,(7,10)_3} \\
 P \setminus \{ q  \} : \quad &    \contraction{1-<$\scriptstyle e_l$>1-<$\scriptstyle e_r$>1-1-1-1-1=<$e$>1-1-1}{(2,5)_2,(4,8)_1,(6,9)_2,(7,10)_3}
   \end{array}
\end{equation*}
\end{example}

\begin{example} Consider the set $X = \{ x_1, x_2, x_3, x_4 \}$, where $x_1 < x_2 < \cdots <  x_4 $.
Let $P$ be a linked pairing with pairs indicated in the diagram below. If  $l = x_2$,  then
$e = \{x_2, x_3 \}$ and $p = \{ x_1,x_3\} \in P$ contains $e$ in its span. In that case we can remove $q = \{ x_1, x_2 \}$
and  $$
 e_{\rm l}   =  \{x_1,x_2\} , \quad e_{\rm r}  =  \{x_2,x_3\}   ,
 $$
 where $e = e_{\rm r}$.
\begin{align*}
\begin{array}{ll}
 P : \quad &    \contraction{1-1=<$e$>1-1}{(1,3)_1,(2,4)_2} \\
 P \setminus \{ q  \} : \quad &    \contraction{1-<$\scriptstyle e_l$>1=<$\scriptstyle e_r$>1-1}{(2,4)_1}
   \end{array}
\end{align*}
\end{example}
\end{proof}

\begin{proof}[Proof of Theorem \ref{thm:mainenergy}]
  We prove the theorem by induction in $n \in \N$.  We make the
  following induction hypothesis.

\vspace{0.5cm}
\noindent
$I_n$: \ \ There are unique numbers $\EE_m$ for $m \in N_n$ such that
the following holds for the functions $C_m, G_m : [0,\infty) \times (0,1] \to \mathcal{L}(\HH_{\rm at})$ defined in
\eqref{eq:defofc} and  \eqref{eq:defofg}.
\begin{enumerate}[label=(\emph{\roman*})]
\item \label{itm:mainenergy:1}
 For $m \in  N_n$ the functions $C_m, G_m$
 are continuous on  $[0,\infty) \times (0,1]$ and bounded and extend to
continuous functions on $[0,\infty) \times [0,1]$.
\item \label{itm:mainenergy:2} For $m \in  N_n$ the functions
$C_m, G_m$ are on $(0,\infty) \times (0,1)$ continuously differentiable
with respect to $r$ and $\eta$ with uniformly bounded derivatives.
\item \label{itm:mainenergy:3} For $m \in  N_n$ we have
$\EE_m = \inn{ \varphi_{\rm at}, G_m(0,0) \varphi_{\rm at} }$.
\end{enumerate}

\vspace{0.5cm}
\noindent
First observe that by definition of $C_n$ and $G_n$ vanish for $n$ odd.
Moreover, we note that $R^{\perp}(r,\eta)$ is continuous on $[0,\infty) \times [0,1)$
whereas  $R^{\parallel}(r,\eta)$ is continuous on  $(0,\infty) \times [0,1)$, with
a discontinuity at $r=0$.

\vspace{0.5cm}
\noindent
For $n=2$,  the Hypothesis $I_n$ can be seen as follows.
We have by definition  for $r \geq 0$ and $\eta > 0$
\begin{align*}
C_2(r,\eta) & = \int G^*(k) \frac{1-P_{\rm at} 1_{|k| + r  = 0 } }{H_{\rm at} - E_{\rm at} + |k| + \eta  + r  } G(k) dk \\
&  =
\int G^*(k) \frac{1 }{H_{\rm at} - E_{\rm at} + |k| + \eta  + r  } G(k) dk ,
\end{align*}
where in the second equality we used that $\{ k \in \R^3 :|k|=0\}$ is
a set of measure zero.
Note that we have $G_2 = C_2$.
\ref{itm:mainenergy:1} follows from dominated convergence (or Lemma \ref{lem:lemestimate}).
\ref{itm:mainenergy:2} follows from the usual results about interchanging integration and differentiation
(or Lemma \ref{lem:lemestimate2}).
\ref{itm:mainenergy:3} follows  from the definition $\EE_2 := \inn{ \varphi_{\rm at}, G_2(0,0) \varphi_{\rm at} }$.

\vspace{0.5cm}
\noindent
Now let us show the induction step. Suppose that $I_n$  holds.
If $n$ is even, the induction hypothesis trivially holds for
$n+1$ since in that case $C_{n+1}, G_{n+1}$ vanish
identically as a direct consequence of the definition.  Thus suppose
$n$ is odd.  By estimating the remainder
 of a first order Taylor expansion, it follows   from the induction hypothesis that
for $m \in N_n$  there exists constants $d_m$ such that
\begin{equation} \label{eq:estonG} | P_{\rm at} \widehat{G}_m(r,\eta)
  P_{\rm at} | \leq d_m | r + \eta | .
\end{equation}
Next we observe that
\begin{subequations} \begin{align}
  \label{eq:estonrespape}
  \| R^\perp(r,\eta) \| &  \leq \frac{1}{\inf ( \sigma(H_{\rm at})
    \setminus \{ E_0 \} ) - E_0 } , \\
  \| R^\parallel(r,\eta) \| & \leq (r+\eta)^{-1} .   \label{eq:estonrespapeb}
\end{align}
\end{subequations}
From the induction
hypothesis \ref{itm:mainenergy:1}, Eq.  \eqref{eq:algebraicenergyhat} of Proposition \ref{thm:algebraicenergy} and    Eqns. \eqref{eq:estonG} and \eqref{eq:estonrespapeb},
we see
that for all $m \in N_n$ there exists a constant $c_m$ such that for
all $r \geq 0,   \eta >   0$ we have
\begin{align} \label{eq:estonthatb}
  \| P_{\rm at}   \widehat{T}_m(r,\eta)    P_{\rm at} \| &\leq    c_m (r+\eta)    ,
  \\
  \| \bar{P}_{\rm at}   \widehat{T}_m(r,\eta)   P_{\rm at} \|  &\leq c_m    ,
  \\
  \| P_{\rm at}   \widehat{T}_m(r,\eta)    \bar{P}_{\rm at} \| &\leq c_m  ,
  \\
  \| \bar{P}_{\rm at}     \widehat{T}_m(r,\eta)    \bar{P}_{\rm at} \| &\leq  c_m  (r+\eta)^{-1}.  \label{eq:estonthate}
\end{align}
Now using the  decomposition of the resolvent \eqref{eq:decofresolv} and the    bounds   \eqref{eq:estonthatb}--\eqref{eq:estonthate} and
 \eqref{eq:estonrespape}
we see that   for $m \in N_n$ there exists a
constant $c_m$ such that
\begin{align*}
  \|  R(r,\eta)  \widehat{T}_m(r,\eta)  R(r,\eta)   \| \leq    \frac{c_m}{r+\eta} .
\end{align*}
Moreover, we see from  \eqref{eq:algebraicenergyhat} and the induction
hypothesis \ref{itm:mainenergy:1}
that the term $ R(r,\eta)  \widehat{T}_m(r,\eta)  R(r,\eta)$ is continuous on
$(0,\infty) \times (0,1]$ and has a continuous extension to $(0,\infty) \times [0,1]$.
 Thus it follows from the definition of $C_n$, given in \eqref{eq:defofc},
 and  Lemma \ref{lem:lemestimate} that $C_{n+1}$ is bounded and has a continuous extension to
  $[0,\infty) \times [0,1]$. Now it follows from \eqref{eq:defofg} that the same holds for $G_{n+1}$.
  Thus we have shown \ref{itm:mainenergy:1} for $n+1$.

From   \eqref{eq:algebraicenergyhat} of  Proposition \ref{thm:algebraicenergy}
and   the induction hypothesis \ref{itm:mainenergy:2} we see  that
$\widehat{T}_n$ is continuously differentiable on $(0,\infty) \times (0,1)$.
Now let  $\xi = r$ or $\xi = \eta$.
Calculating the derivative using the product rule, we obtain similarly as before,
with    Eq.  \eqref{eq:estonG} and
\begin{equation*}
\| \partial_\xi R^\perp(r,\eta) \| \leq \frac{1}{(\inf ( \sigma(H_{\rm at})
    \setminus \{ E_0 \} ) - E_0 )^2}
 , \quad \| \partial_\xi R^\parallel(r,\eta) \| \leq
(r+\eta)^{-2} ,
\end{equation*}
  the bounds
\begin{align}\label{eq:estonderthatb}
  \| P_{\rm at} \partial_\xi   \widehat{T}_n(r,\eta)    P_{\rm at} \| &\leq    c_m   ,
  \\
  \| \bar{P}_{\rm at} \partial_\xi   \widehat{T}_n(r,\eta)   P_{\rm at} \|  &\leq {c_m}{(r+\eta)}^{-1}     , \label{eq:estonderthatc}
  \\
  \| P_{\rm at}  \partial_\xi  \widehat{T}_n(r,\eta)    \bar{P}_{\rm at} \| &\leq {c_m}{(r+\eta )}^{-1} ,\label{eq:estonderthatd}
  \\
  \| \bar{P}_{\rm at} \partial_\xi     \widehat{T}_n(r,\eta)    \bar{P}_{\rm at} \| &\leq {c_m}{(r+\eta)}^{-2} . \label{eq:estonderthate}
\end{align}
Now  using  \eqref{eq:estonderthatb}--\eqref{eq:estonderthate}  we
obtain for $r > 0$ and $\eta > 0$ the bound
\begin{align*}
  \|  \partial_\xi  R(r,\eta)  \widehat{T}_n(r,\eta) R(r,\eta)   \|  \leq \frac{C}{(r+\eta)^2 }.
\end{align*}
Now we see from the definition of $C_n$,  \eqref{eq:defofc},
and  Lemma \ref{lem:lemestimate2}, that $C_{n+1}$ is continuously
differentiable on $(0,\infty) \times (0,1)$
with uniformly bounded derivatives.  Hence  it follows from \eqref{eq:defofg} that the same holds for $G_{n+1}$.
Thus we have shown \ref{itm:mainenergy:2} for $n+1$.
Property \ref{itm:mainenergy:3} now follows from  the  definition
$\EE_{n+1} := \inn{ \varphi_{\rm at}, G_{n+1}(0,0) \varphi_{\rm at}
}$.  Thus we have shown $I_{n+1}$.

\vspace{0.5cm}
\noindent
Using   \eqref{eq:algebraicenergy} of  Proposition  \ref{thm:algebraicenergy}, it follows
from  \eqref{eq:estonG}  and  \eqref{eq:estonrespape} that
\begin{equation*}
 \lim_{\eta \downarrow 0}
 \inn{ \varphi_{\rm at}, T_{n}(0,\eta) \varphi_{\rm at} } =  \inn{ \varphi_{\rm at}, G_{n}(0,0) \varphi_{\rm at} } =  \EE_{n} ,
\end{equation*}
where the last equality follows from \ref{itm:mainenergy:3} of the induction hypothesis.
Setting   $E_n := \EE_n$,  the claim of the theorem follows from Lemma \ref{lem:energyform2}.
\end{proof}

As a byproduct of the proof we obtain the following corollary, which
tells us that we can calculate the coefficients $E_n$ solely in terms
of linked pair  partitions.

\begin{corollary} \label{cor:energyformula} Let the situation be as in
  Theorem \ref{thm:mainenergy}.  Then we have
  \begin{align*}
    E_n  =  \lim_{\eta \downarrow 0} \langle \varphi_{\rm at} ,
    \Bigg\{ & C_n(0,\eta)
    \\ &\quad + \sum_{k=2}^n \sum_{\substack{j_1+\ldots + j_k = n\\ j_i \geq
    1}}  \left[ \prod_{i=1}^{k-1}\bigl( \widehat{C}_{j_i}(0,\eta)
    R^\perp(0,\eta) \bigr)\right] \widehat{C}_{j_k}(0,\eta) \Bigg\}
    \varphi_{\rm at} \rangle .
  \end{align*}
\end{corollary}

\section{Ground State}\label{sec:groundstate}

In this section we prove the following theorem,
which shows the existence of the expansion coefficients for the ground state.
The strategy of the proof is analogous to that of Theorem~\ref{thm:mainenergy}, with the difference that one has to
account for the   square of the  resolvent which may now appear in operator products. For an outline of the
proof, we therefore refer  the reader to the outline of the proof of Theorem~\ref{thm:mainenergy}, given at the
beginning of Section \ref{sec:groundstatee}.

\begin{theorem} \label{thm:groundstate} Suppose the assumptions of  Theorem \ref{thm:mainenergy} hold and
let  $(E_n)_{n \in \N}$
  be the unique sequence given in  Theorem \ref{thm:mainenergy}.  Let
  \begin{equation*}
    \psi_0 = \varphi_{\rm at} \otimes \Omega .
  \end{equation*}
  Then for all $m \in \N$ the following limit exists
  \begin{equation*}
    \psi_m = \lim_{\eta \downarrow 0}  \psi_m(\eta) ,
  \end{equation*}
  where
  \begin{equation}
    \label{eq:E_neegs2}
    \psi_m(\eta) := \sum_{k=1}^m \sum_{\substack{j_1 + \cdots + j_k = m \\ j_s \geq 1 }}
    \prod_{s=1}^k  \left\{ (E_0 - H_0 - \eta)^{-1} \bar{P}_0   ( \delta_{1 j_s }  V -  E_{j_s}   ) \right\} \psi_0  .
  \end{equation}
\end{theorem}

To show that the expansion coefficients of the ground state exist, we
have to calculate their norm. For this we introduce the following
graph functions.
For  $m,n  \in \Z$ with $m \leq n$ we define the set
$$
N_{m,n} = [m,n] \cap \Z \setminus \{ 0 \} ,
$$
and
$$
\pi_{m,n}(r,\eta) = ( (V_x)_{x \in N_{m,n} } , (\widetilde{R}_e(\cdot + r,\eta))_{e
  \in E_{N_{m,n}} } ) ,
$$
where for $e \in E_{N_n}$ we defined
$$
\widetilde{R}_e(r,\eta) = \begin{cases}
  R(r,\eta)^2  , &  \text{ if } e = \{-1,1\}   \\ R(r,\eta) ,
  &  \text{ otherwise.}
\end{cases}
$$

\begin{example}
  We can write $\pi_{-3,2}(r,\eta)$ symbolically as
  
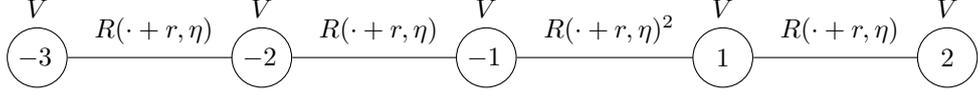
\begin{footnotesize}
  \begin{equation*}
    \bignode{\! -3}{V}
    \edge{R(\cdot + r , \eta)   }
    \bignode{\! -2}{V}
    \edge{R(\cdot + r , \eta)   }
    \bignode{\! -1}{V}
    \edge{R(\cdot + r , \eta)^2   }
    \bignode{1}{V}
    \edge{R(\cdot + r , \eta)   }
    \bignode{2}{V}
  \end{equation*}
  \end{footnotesize}
\end{example}
Note that for $n \in \N$ we have $\pi_{1,n}(r,\eta) = \pi_n(r,\eta)$.
 For  $M \subset \Z$ we define the set $\mathcal{Q}_0(M)$ consisting of all collections
 of disjoint nonempty intervals of $M$, such that 0 does not lie
 between the endpoints of any of   the intervals, i.e.,
\begin{align*}
  \mathcal{Q}_0(M) := \{ {}& \mathcal{I} \subset \mathcal{P}(M)   :  \ \forall I , J \in \mathcal{I} \text{ we have }
  I \cap J = \emptyset ,  \\
                           & \text{if } I \in \mathcal{I}, \text{ then } 0 \notin [ \min I , \max I ],
  \\
                           & \forall I \in \mathcal{I} \text{ the set } I \text{ is a nonempty interval of   } M \    \} .
\end{align*}
Note that $\mathcal{Q}_0(M) \subset \mathcal{Q}(M)$.  We define
\begin{align}
  \label{eq:E_n:2gs}
  T_{m,n}(r,\eta)   := & \,  P_\Omega \Pi ( \pi_{m,n}(r,\eta)) P_\Omega
                            +  \sum_{\substack{ \mathcal{I} \in  \mathcal{Q}_0(N_{m,n}) : \\ |\mathcal{I}| \geq 1  }}
							P_\Omega  \Pi (\subst_{\substack{I  \to -E_{|I|}\\ I \in \mathcal{I} }}(\pi_{m,n}(r, \eta ))) P_\Omega
					\nonumber \\
                    =& \sum_{\substack{ \mathcal{I} \in  \mathcal{Q}_0(N_{m,n})  }}
                      P_\Omega  \Pi (\subst_{\substack{I  \to -E_{|I|}\\ I \in \mathcal{I} }}(\pi_{m,n}(r, \eta))) P_\Omega ,
\end{align}
as an operator on the atomic Hilbert space.  As an
immediate consequence of the definitions we obtain the following
lemma.  To be explicit we give a proof below.

\begin{lemma}
\label{lem:eqwickgs}  Suppose the assumptions  of Theorem \ref{thm:groundstate} hold.
 Then with the definition  \eqref{eq:E_neegs2} we
  have for all $m \in \N$ that
$$
\inn{ \psi_{m}(\eta), \psi_{m}(\eta) } = \inn{ \varphi_{\rm at} ,
  T_{-m,m}(0,\eta) \varphi_{\rm at} } .
$$
\end{lemma}
\begin{proof} The proof is analogous to the proof of Lemma \ref{lem:energyform2}.
 Inserting \eqref{eq:E_neegs2} into the left hand side and taking the adjoint we find
  \begin{align}
      \langle \psi_m(\eta) ,  \psi_m&(\eta)\rangle =
    \sum_{k'=1}^m  \sum_{k=1}^m  \sum_{\substack{j_1' + \cdots + j_{k'}' = m \\ j_s' \geq 1 }}  \sum_{\substack{j_1 + \cdots + j_k = m \\
        j_s \geq 1 }} \nonumber \\
  &\langle \psi_0 ,  \prod_{s'=1}^{k'} \left\{  ( \delta_{1 j'_{s'} }  V -
        E_{j'_{s'}}   )  (E_0 - H_0 - \eta)^{-1} \bar{P}_0  \right\} \nonumber\\
     &\qquad\qquad\quad\prod_{s=1}^k  \left\{ (E_0 - H_0 - \eta)^{-1} \bar{P}_0   (
        \delta_{1 j_s }  V -  E_{j_s}   ) \right\} \psi_0 \rangle. \label{eq:combgs}
  \end{align}
  Consider the summand in \eqref{eq:E_n:2gs} indexed by
  $\mathcal{I}  \in \mathcal{Q}_0(N_{-m,m})$.
   We partition the  set $\mathcal{I}$ into
  $\mathcal{I}_1 = \{ I \in \mathcal{I} :   I \subset N_{-m,-1}    \}$ and
   $\mathcal{I}_2 = \{ I \in \mathcal{I} : I \subset N_{1,m}  \}$. By definition
   of  $\mathcal{Q}_0(N_{-m,m})$ this is indeed a partition of $\mathcal{I}$.
   As in the proof of Lemma \ref{lem:energyform2}
we define
\begin{align*}
\mathcal{K}_1 & := \{ \{ s \} : s \in N_{-m,-1}  \text{ and } s \notin I =
\emptyset , \ \forall I \in \mathcal{I}_1 \}, \\
 \mathcal{K}_2 & := \{ \{ s \} : s \in N_{1,m}  \text{ and } s \notin I =
\emptyset , \ \forall I \in \mathcal{I}_2 \} .
\end{align*}
Now we order the elements of
$\mathcal{S}_j := \mathcal{I}_j \cup \mathcal{K}_j$ in increasing order as in the proof of Lemma \ref{lem:energyform2}.
This defines a bijection $\varphi_j: N_{|\mathcal{S}_j|} \to \mathcal{S}_j$ preserving the order.
 By construction we see that the summand in
\eqref{eq:E_n:2gs} indexed by $\mathcal{I}$ is equal to the
summand in \eqref{eq:combgs} which we obtain  by choosing $k' = |\mathcal{S}_1|$ and  $k = |\mathcal{S}_2|$,
$j'_{s'} = |\varphi_1(s')|$ and $j_s = |\varphi_2(s)|$, by choosing $-E_1$ in case $j'_{s'}=1$ and $\varphi_1(s')  \in \mathcal{I}_1$ or  $j_s=1$ and $\varphi_2(s)  \in \mathcal{I}_2$, and
by choosing   $V$  in case $j'_{s'}=1$ and $\varphi_1(s')  \notin \mathcal{I}_1$ or  $j_s=1$ and $\varphi_2(s)  \notin \mathcal{I}_2$.
\end{proof}

For $m,n \in \Z$ with $m \leq n$  we define
\begin{align*}
  C_{m,n}(r,\eta)  &:=   \sum_{\substack{ P_e \text{ linked pairing of } N_{m,n} \\ S(P_e) =  [m,n]      } }
  \mathcal{C}^0_{P_e}( \subst_{\substack{I  \to \widetilde{T}_{I}(\cdot + r,\eta)    \\ I \in \mathcal{I}_{P_e}}} (
  \pi_{m,n}(r,\eta)))
\end{align*}
where
$$
\widetilde{T}_{I}(r,\eta) :=
\begin{cases}
  {T}_{\min I,\max I}(r,\eta) - E_{|I|+1}    , & \text{ if } 0 \notin [\min I  , \max I ]  \\
  {T}_{\min I,\max I}(r,\eta) , & \text{
    otherwise.}
\end{cases}
$$

Observe that if $m,n \in \Z$ have the same sign and $m \leq n$, then
$$
T_{m,n} = T_{n-m+1}  , \quad C_{m,n} = C_{n-m+1} .
$$

\begin{proposition}\label{thm:algebraicgs}
For $m,n \in \Z$ with $m \leq n$ we have
  \begin{align} \label{eq:eqfortmn1}
    &T_{m,n}(r,\eta)  = C_{m,n}(r,\eta)  \nonumber \\
     &\quad\;+  \sum_{k=2}^{n-m}
     \sum_{\substack{j_1+\ldots + j_k = n-m\\ j_i \geq 1}}
    \left[ \prod_{i=1}^{k-1}\bigl( \widetilde{C}_{l_i(m,\underline{j}),r_i(m,\underline{j})}(r,\eta) \widetilde{R}_{\{r_i(m,\underline{j}) ,l_{i+1}(m,\underline{j}) \}}(r,\eta) \bigr)
    \right] \nonumber \\
    & \qquad\qquad \times \widetilde{C}_{l_{k}(m,\underline{j})  ,r_k(m,\underline{j})  }(r,\eta)  ,
  \end{align}
  where we defined inductively
for $\underline{j} = (j_1,\dotsc,j_k)$ the numbers   $l_1(m,\underline{j}):= m$ and
  $l_{i+1}(m,\underline{j}) :=  r_{N_{m,n}}^{j_i}(l_i(m,\underline{j}))$,
  and we defined  $r_i(m,\underline{j}) := r_{N_{m,n}}^{j_i-1}(l_i(m,\underline{j}))$  and
  $$
\widetilde{C}_{p,q}(r,\eta) :=
\begin{cases}
  {C}_{p,q}(r,\eta)   , & \text{ if } 0 \in [p,q] \\
  {C}_{p,q}(r,\eta) - E_{q-p+1}  , & \text{otherwise. }
\end{cases}
$$
\end{proposition}

The proof is very similar to that of Proposition
\ref{thm:algebraicenergies}, except we have to consider the factor
involving the square of the resolvent and the fact that we have less
energy subtractions.

\begin{proof} The case, where  $m,n$ have the same sign, has already been shown in the
last section. Thus assume $m < 0 < n$.
  By the generalized Wick theorem,  Lemma \ref{lem:genwick}, we have
  \begin{align*}
    T_{m,n}(r,\eta)
    &=\sum_{P \text{ pairing of } N_{m,n}   }   \sum_{\substack{\mathcal{I} \in  \mathcal{Q}_0(N_{m,n})    }}
    \mathcal{C}^0_P (\subst_{\substack{I \to -\EE_{|I|}\\ I \in \mathcal{I }} }(\pi_{m,n}(r,\eta))) \\
    &=  T_{m,n}^{(C)}(r,\eta) +  T_{m,n}^{(D)}(r,\eta) ,
  \end{align*}
  where we divided the sum over the partitions into partitions which
  connect the smallest and the largest vertex
  \begin{align*}
    T_{m,n}^{(C)}(r,\eta) :=   \sum_{\substack{ P \text{ pairing of } N_{m,n}  \\ S(P) =  [m,n]   }}   \sum_{\substack{\mathcal{I} \in  \mathcal{Q}_0(N_{n,n})     }}
    \mathcal{C}^0_P  (\subst_{\substack{I \to -\EE_{|I|}\\ I \in \mathcal{I }} }(\pi_{m,n}(r,\eta)))
  \end{align*}
   and
  partitions which are disconnected
  \begin{align*}
    T_{m,n}^{(D)}(r,\eta)
    :=  \sum_{\substack{ P \text{ pairing of } N_{m,n}  \\ S(P) \neq [m,n]   }}   \sum_{\substack{\mathcal{I} \in  \mathcal{Q}_0(N_{m,n})  }}
    \mathcal{C}^0_P  (\subst_{\substack{I \to -\EE_{|I|}\\ I \in \mathcal{I }} }(\pi_{m,n}(r,\eta)))
  \end{align*}
  To simplify the connected part $T_n^{(C)}(r,\eta)$ we use the
  decomposition \eqref{eq:ppedcomp} as in the proof of
  Proposition~\ref{thm:algebraicenergies} and an analogous argument
  yields
  \begin{align*}
    T_{m,n}^{(C)}(r,\eta) &=\sum_{\substack{  P_e  \text{ pairing of } N_{m,n}  \\ \{ 1,n  \}  \subset  \bigcup P_e    \\  P_e \text{ linked}   } }
    \prod_{I \in   \mathcal{I}_{P_e} } \bigg\{  \sum_{\substack{P_I
    \text{ pairing  of }   I } } \bigg\} \\
    &\qquad\qquad\qquad
    \sum_{\substack{\mathcal{I} \in  \mathcal{Q}_0(N_{m,n})    }}
    \mathcal{C}^0_{(P_e \cup \bigcup_{I \in \mathcal{I}_{P_e}}  P_I) } (\subst_{\substack{J \to -\EE_{|J|}\\ J \in \mathcal{I}} }(\pi_{m,n}(r,\eta)))  \nonumber \\
    &= \sum_{\substack{  P_e  \text{ pairing of } N_n  \\ \{ 1,n  \}
		\subset  \bigcup P_e    \\  P_e \text{ linked }   }  }
    \prod_{I \in   \mathcal{I}_{P_e} }
		\bigg\{ \sum_{\substack{\mathcal{I}_I  \in  \mathcal{Q}_0(I)}}
		\sum_{\substack{P_I  \text{ pairing} \\
    \text{of }  I \setminus \bigcup  \mathcal{I}_I} } \bigg\} \\
   &\qquad\qquad\qquad\qquad
   \mathcal{C}^0_{(P_e \cup \bigcup_{I \in \mathcal{I}_{P_e}}  P_I)  }
   (\subst_{\substack{J \to -\EE_{|J|}\\
	J \in \bigcup_{I \in \mathcal{I}_{P_e}}
	\mathcal{I}_I} }(\pi_{m,n}(r,\eta))) \nonumber \\
   &=  \sum_{\substack{  P_e  \text{ pairing of } N_n  \\ \{ 1,n  \}
		\subset  \bigcup P_e    \\  P_e \text{  linked}   } }
	\mathcal{C}^0_{P_e}( \subst_{\substack{I  \to \widetilde{T}_{I}(\cdot + r,\eta)
	\\ I \in \mathcal{I}_{P_e}}} ( \pi_{m,n}(r,\eta)))  \nonumber \\
  &= C_{m,n}(r,\eta) .
  \end{align*}

  To simplify the disconnected part $T_{m,n}^{(D)}(r,\eta)$ we
  rearrange the sum as in the proof of Proposition \ref{thm:algebraicenergies},
  which yields   the following identities,
  \begin{align*}
    &T_{m,n}^{(D)}(r,\eta) \nonumber \\
    &= \!\!\! \sum_{\substack{ \mathcal{K} \in \mathcal{Q}(N_{m,n}) \\  \bigcup \mathcal{K} =  N_{m,n} \\ |\mathcal{K}| \geq 2}} \prod_{K \in \mathcal{K}}
    \!\bigg\{\! \sum_{\substack{ P_K \text{ pairing of } K \\  S(P_K) = \\  [\min  K , \max K]   }} \sum_{\substack{\mathcal{I}_K  \in  \mathcal{Q}_0(K)     }}
    +    \ 1_{(P_K = \emptyset , \mathcal{I}_K = \{ K \} , 0 \notin   [\min  K , \max K] ) }  \bigg\}  \nonumber \\
    & \qquad\qquad\qquad\qquad\qquad\times \ \mathcal{C}^0_{(\bigcup_{K \in \mathcal{K}}  P_K)}( \subst_{\substack{I \to -\EE_{|I|}\\ I \in  \bigcup_{K \in \mathcal{K} } \mathcal{I}_K }}(\pi_{m,n}(r,\eta)))
     \nonumber \\
    &= \!\!\!   \sum_{\substack{ \mathcal{K} \in \mathcal{Q}(N_{m,n})
    \\  \bigcup \mathcal{K} =  N_{m,n} \\ |\mathcal{K}| \geq 2}}
    \prod_{K \in \mathcal{K}}\!
    \bigg\{ \! \sum_{\substack{ P_K  \text{ pairing of } K \\  S(P_K) = \\  [\min  K , \max K]   }} \sum_{\substack{\mathcal{I}_K  \in  \mathcal{Q}_0(K)     }}
    +    \ 1_{(P_K = \emptyset , \mathcal{I}_K = \{ K \} , 0 \notin   [\min  K , \max K]  ) } \bigg\}  \nonumber  \\
    & \qquad\quad \prod_{K  \in \mathcal{K} \setminus \max \mathcal{K} }\! \bigg\{  \mathcal{C}^0_{P_K}  (\subst_{\substack{I \to -\EE_{|I|}\\ I \in \mathcal{I}_K }}(\pi_{\min K, \max K }(r,\eta)))  [R(r,\eta)]^{1+1_{0 \in [\min K , \max K  ]}} \bigg\}  \nonumber \\
    & \qquad\qquad\qquad\qquad\qquad\times \ \mathcal{C}^0_{P_{\max \mathcal{K}} }
    (\subst_{\substack{I \to -\EE_{|I|}\\ I \in \mathcal{I}_{\max \mathcal{K}} }}(\pi_{\min (\max \mathcal{K}), \max (\max \mathcal{K}) }(r,\eta)))
    \\
 & =   \sum_{k=2}^{n-m} \sum_{\substack{j_1+\ldots + j_k = n-m\\ j_i \geq
    1}} \left[ \prod_{i=1}^{k-1}\bigl( \widetilde{C}_{l_i(m,\underline{j})  ,r_i(m,\underline{j})  }(r,\eta) \widetilde{R}_{\{r_i(m,\underline{j})  ,l_{i+1}(m,\underline{j})  \}}(r,\eta) \bigr)
    \right] \\
    &\qquad\qquad\qquad\qquad\qquad\qquad\times \widetilde{C}_{l_{k}(m,\underline{j})  ,r_k(m,\underline{j})  }(r,\eta) ,
  \end{align*}
  where we ordered $\mathcal{K}$  with respect to the ordering defined in
  \eqref{eq:orderingintervals}, and in the last equality we identified the
  summation indices as   follows:   $k = |\mathcal{K}|$ and  $j_i = |{K}_i|$,  for $\mathcal{K} = \{K_1, K_2, \cdots , K_k \}$ with
  $K_1 <  K_2 < \cdots < K_k$.
\end{proof}

As in Subsection \ref{sec:algebraicform},
 we want to decompose the resolvent occurring in \eqref{eq:eqfortmn1}
according to  \eqref{eq:decofresolv}.
To this end we define
\begin{align}  \label{eq:defofgtilde}
G_{m,n}&(r,\eta) :=  \nonumber   \\ &C_{m,n}(r, \eta )
   +  \sum_{k=2}^{n-m}
  \sum_{\substack{j_1+\ldots + j_k = n-m\\ j_i \geq 1}} \nonumber   \\
   &\qquad\qquad\qquad\quad \left[ \prod_{i=1}^{k-1}\bigl( \widetilde{C}_{l_i(m,\underline{j})  ,r_i(m,\underline{j})  }(r,\eta) P^\perp \widetilde{R}_{\{r_i(m,\underline{j}) ,l_{i+1}(m,\underline{j}) \}}(r,\eta) \bigr)
    \right] \nonumber \\
  & \qquad\qquad\qquad\qquad\qquad \times \widetilde{C}_{l_{k}(m,\underline{j}) ,r_k(m,\underline{j})  }(r,\eta)  ,
\end{align}
where
$$
P^\perp := 1-P_{\rm at} .
$$
Moreover, we define
  $$
\widetilde{G}_{p,q}(r,\eta) :=
\begin{cases}
  {G}_{p,q}(r,\eta)   , & \text{ if } 0 \in [p,q] \\
  {G}_{p,q}(r,\eta) - E_{q-p+1} , & \text{ otherwise. }
\end{cases}
$$
and
$$
P^\parallel := P_{\rm at} .
$$

\begin{theorem}\label{thm:groundstateform}
  We have
  \begin{align}\label{eq:algebraicgsexp}
    T_{m,n}&(r, \eta)  = \nonumber \\
    & G_{m,n}(r, \eta)
    +  \sum_{k=2}^{n-m} \sum_{\substack{j_1+\ldots + j_k = n-m\\ j_i \geq
    1}} \nonumber \\
    &\qquad\qquad\qquad\quad
    \left[ \prod_{i=1}^{k-1}\bigl( \widetilde{G}_{l_i(m,\underline{j}) ,r_i(m,\underline{j}) }(r,\eta) P^\parallel \widetilde{R}_{\{r_i(m,\underline{j}) ,l_{i+1}(m,\underline{j}) \}}(r,\eta) \bigr)
    \right] \nonumber \\
    &\qquad\qquad\qquad\qquad\qquad\times \widetilde{G}_{l_{k}(m,\underline{j}) ,r_k(m,\underline{j}) }(r,\eta)  .
  \end{align}
\end{theorem}

\begin{proof}
We start with the formula in Proposition
  \ref{thm:algebraicgs} and write  the resolvent as a sum of
  parallel and orthogonal part.  Then we
 multiply out the resulting expression
 and, as in the proof of Proposition  \ref{thm:algebraicenergy}, we  collect
 the terms according to the
  number, $s-1$, of times
  $P^\parallel$ occurs.
 Starting with the second term of the right hand side of
 \eqref{eq:eqfortmn1} we obtain  by straight forward algebraic calculation
  \begin{align} 
    \label{eq:E_n:2332}
      &\sum_{k=2}^{n-m} \sum_{\substack{j_1+\ldots + j_k = n-m\\ j_i \geq
    1}} \left[ \prod_{i=1}^{k-1}\bigl( \widetilde{C}_{l_i(m,\underline{j}) ,r_i(m,\underline{j}) }(r,\eta) \widetilde{R}_{\{r_i(m,\underline{j}) ,l_{i+1} (m,\underline{j})\}}(r,\eta) \bigr)
    \right] \nonumber \\[-5pt]
    &\qquad\qquad\qquad\qquad\qquad \times \widetilde{C}_{l_{k}(m,\underline{j}) ,r_k(m,\underline{j}) }(r,\eta)
  \nonumber \\
    &\quad =
   \sum_{k=2}^{n-m} \sum_{\sigma_1,\dotsc,\sigma_{k-1} \in \{ \perp, \parallel \}}
   \sum_{\substack{j_1+\ldots + j_k = n-m \\ j_i \geq 1}} \nonumber \\[-4pt]
    &\qquad\qquad \left[ \prod_{i=1}^{k-1}\bigl( \widetilde{C}_{l_i(m,\underline{j}) ,r_i(m,\underline{j}) }(r,\eta) P^{\sigma_i} \widetilde{R}_{\{r_i(m,\underline{j}) ,l_{i+1}(m,\underline{j}) \}}(r,\eta) \bigr)
    \right] \nonumber \\
    & \qquad\qquad\qquad \times \widetilde{C}_{l_{k}(m,\underline{j}) ,r_k(m,\underline{j}) }(r,\eta)
     \nonumber \\
    &\quad =
    \sum_{s=1}^{n-m}   \sum_{\substack{n_1+\ldots + n_s = n-m \\ n_i \geq   1}}
    \sum_{\substack{ k_1,\dotsc, \, k_s \in \N  \\ k_1 + \cdots + k_s \geq 2 } }
     \sum_{\substack{j_{1,1} +\ldots + j_{1,k_1}  = n_1 \\ j_{1,i}  \geq  1}} \cdots
     \sum_{\substack{j_{s,1} +\ldots + j_{s,k_s}  = n_s \\ j_{s,i}  \geq  1}}
     \nonumber \\
    &\qquad\qquad    \left[ \prod_{i_1=1}^{k_1-1}\bigl(  \widetilde{C}_{l_{i_1}(m,\underline{j}_1), r_{i_1}(m,\underline{j}_1)}
    P^{\perp} \widetilde{R}_{\{r_{i_1}(m,\underline{j}_1), l_{i_1+1}(m,\underline{j}_1) \}   }(r,\eta) \bigr)
    \right] \nonumber \\
    &\qquad\qquad\qquad\quad \times \widetilde{C}_{l_{k_1}(m,\underline{j}_1), r_{k_1}(m,\underline{j}_1)} P^{\parallel} \widetilde{R}_{r_{k_1}(m,\underline{j}_1), l_{1}(l_1(m,\underline{n}),\underline{j}_2)} \nonumber \\
    & \qquad\qquad \cdots    \nonumber \\[-4pt]
    &\qquad\qquad
         \Bigg[ \prod_{i_s=1}^{k_s-1}\bigl(  \widetilde{C}_{l_{i_s}( l_{s-1}(m,\underline{n}),\underline{j}_s), r_{i_s}(l_{s-1}(m,\underline{n}),\underline{j}_s)}
 P^\perp   \nonumber \\[-4pt]
 &\qquad\qquad\qquad\qquad\qquad\quad \widetilde{R}_{\{r_{i_s}(l_{s-1}(m,\underline{n}),\underline{j}_s), l_{i_s+1}(l_{s-1}(m,\underline{n}),\underline{j}_s) \}   }(r,\eta) \bigr)
    \Bigg]   \nonumber \\[-3pt]
    & \qquad\qquad\qquad\quad \times \widetilde{C}_{l_{k_s}(l_{s-1}(m,\underline{n}),\underline{j}_s), r_{k_s}(l_{s-1}(m,\underline{n}),\underline{j}_s)} \nonumber \\
    &\quad =
    \sum_{k=2}^{n-m} \sum_{\substack{j_1+\ldots + j_k = n-m\\ j_i \geq
    1}} \nonumber \\
    &\qquad\qquad\qquad \left[ \prod_{i=1}^{k-1}\bigl( \widetilde{C}_{l_i(m,\underline{j})  ,r_i(m,\underline{j})  }(r,\eta) P^\perp \widetilde{R}_{\{r_i(m,\underline{j}) ,l_{i+1}(m,\underline{j}) \}}(r,\eta) \bigr)
    \right] \nonumber \\
    &\qquad\qquad\qquad\qquad \times \widetilde{C}_{l_{k}(m,\underline{j}) ,r_k(m,\underline{j})  }(r,\eta)\nonumber  \\
    & \qquad\qquad +
		\sum_{s=2}^{n-m}
		\sum_{\substack{n_1+\ldots + n_s = n-m\\ n_i \geq 1}}
		\nonumber \\
	&\qquad\qquad\qquad \qquad
	\left[ \prod_{i=1}^{s-1}
	\bigl( \widetilde{G}_{l_i(m,\underline{n}) ,r_i(m,\underline{n}) }(r,\eta)
	P^\parallel
	\widetilde{R}_{\{r_i(m,\underline{n}) ,l_{i+1}(m,\underline{n})\}}(r,\eta)
	\bigr)\right] \nonumber \\
    &\qquad\qquad\qquad\qquad\qquad\qquad\times \widetilde{G}_{l_{s}(m,\underline{n}) ,r_k(m,\underline{n}) }(r,\eta)  ,
  \end{align}
  where the first term on the very right hand side, originates from $s=1$, and for the second
  term we used the definition in  \eqref{eq:defofgtilde}.
  Adding $C_{m,n}(r,\eta)$ to the right hand side  of \eqref{eq:E_n:2332}
   and using  again the definition in    \eqref{eq:defofgtilde}, the theorem follows.
\end{proof}

\begin{lemma}\label{lem:lemestimate3}
  Let $X \subset \Z$ be  a finite set and   let $S \subset \R^d$.
  Suppose for each $s \in S$ we  are given
 $$
 \phi_s = ( (V)_{x \in X} , (F_{e,s})_{e \in {E}_X} ) ,
 $$
  a  $(\mathcal{V}_{\rm sb}, \mathcal{R}_{\rm sb})$-valued  graph function  on $X$.
  Suppose there exists a   constant $C_F$  and an $e' \in E_X$  such
  that
   \begin{align*}
   \|F_{e,s}(r) \|  & \leq C_F (|r|^{-1} + 1 ) , \quad  \forall  r  >    0  , \quad  e \in {E}_X \setminus \{ e' \}  ,  \\
   \|F_{e',s}(r) \| & \leq C_F (|r|^{-1} + 1 )^2 , \quad  \forall r  >    0  .
   \end{align*}
   Suppose that for every $r > 0$ and $e \in E_X$ the function $s \mapsto  F_{e,s}(r)$ is continuous.
  Then for any linked pair partition $P$ of $X$ the function $s \mapsto  \CC_P(\phi_s)(r)$
  is continuous  for each $r \geq  0$ and
  $$
\| \CC_P(\phi)(r)  \| \leq C_F^{|X|} C_1^{|X|-1} ,
$$
where $C_1$ is defined in Lemma \ref{lem:lemestimate}.
\end{lemma}
\begin{proof} The estimate follows analogous as the estimate in the proof of Lemma \ref{lem:lemestimate}.
The statement about the continuity follows from dominated convergence.
\end{proof}

\begin{proof}[Proof of Theorem \ref{thm:groundstate}]
  From the proof of Theorem \ref{thm:mainenergy} we know various
  properties about $C_n$, $G_n$, and $T_n$, and respectively
  $\widehat{C}_n$ and $\widehat{G}_n$ and $\widehat{T}_n$.
We make the following hypothesis:

\vspace{0.5cm}
\noindent
$J_n$: \ \ For $m_1, m_2  \in N_n$ the function $C_{-m_1,m_2}(r,\eta)$
is continuous and uniformly bounded on $[0,\infty) \times (0,1]$. Moreover it extends to a continuous
function on  $[0,\infty) \times [0,1]$.

\vspace{0.5cm}
\noindent
$J_1$ holds,  since $C_{-1,1} = C_2$.

\noindent
Next we show  the induction step $n \to n+1$.
Thus suppose that $J_{n}$ holds. For all  $m_1, m_2  \in N_n$  it follows from
the definition that the function  $G_{-m_1,m_2}$
is a continuous  uniformly bounded function on $[0,\infty) \times (0,1]$ and extends to a continuous
functions on  $[0,\infty) \times [0,1]$.
Let $m_1, m_2  \in N_n$.
Eq. \eqref{eq:algebraicgsexp}
in Theorem \ref{thm:groundstateform} implies
that  $ T_{m_1,m_2}(r, \eta) $ is a continuous function on $(0,\infty) \times [0,1]$
and satisfies the following bounds.
(Note that there is
either at most one $\widetilde{G}_{p,q}$ with a $0 \in [\min p, \max q]$ or at most one
$\widetilde{R}^\parallel_{\{p,q\}}$ with $0 \in [\min p, \max q]$.)
There exists a
constant $c_n$ such that for all $r > 0$, $\eta \geq  0$ we have
\begin{align}
  \| P_{\rm at}   {T}_{-m_1,m_2}(r,\eta)    P_{\rm at} \| &\leq    c_n,  \label{eq:ptpest}
  \\
  \| \bar{P}_{\rm at}   {T}_{-m_1,m_2}(r,\eta)   P_{\rm at} \|  &\leq c_n  (r+\eta)^{-1}   ,
  \\
  \| P_{\rm at}   {T}_{-m_1,m_2}(r,\eta)    \bar{P}_{\rm at} \| &\leq c_n (r+\eta)^{-1} ,
  \\
  \| \bar{P}_{\rm at}     {T}_{-m_1,m_2}(r,\eta)    \bar{P}_{\rm at} \| &\leq  c_n  (r+\eta)^{-2},  \label{eq:ptpest2}
\end{align}
where we made use of the estimates in the  proof of Theorem \ref{thm:mainenergy}.
Using the  bounds  \eqref{eq:ptpest}--\eqref{eq:ptpest2} we see that for $m_1,m_2  \in  N_n$ there
exists a constant $C$ such that
\begin{align*}
  \|  R(r,\eta)  {T}_{-m_1,m_2}(r,\eta)  R(r,\eta) \|
	\leq  \frac{C}{(r+\eta)^2 }.
\end{align*}
We  conclude  from Lemma \ref{lem:lemestimate3} that $J_{n+1}$
holds.

\vspace{0.15cm} 
\noindent
Knowing that  $J_n$ holds  the definition given in  \eqref{eq:defofgtilde}
implies that $G_{-m_1,m_2}$ has a continuous extension to $[0,\infty) \times [0,1]$.
By Lemma \ref{lem:eqwickgs} and \eqref{eq:algebraicgsexp}  we see that
$$\inn{\psi_m(\eta),\psi_m(\eta) } =  \inn{ \varphi_{\rm at} , T_{-m,m}(0,\eta) \varphi_{\rm at} } =
 \inn{ \varphi_{\rm at} , G_{-m,m}(0,\eta) \varphi_{\rm at} } ,
$$
for which the limit  $\eta \downarrow 0$  exists (observe that the second term in \eqref{eq:algebraicgsexp}
does not contribute, since $R$ contains the projection onto the
complement of the unperturbed ground state).

\vspace{0.15cm} 
\noindent
Finally, we will show that  the convergence of $\psi_m(\eta)$, as $\eta \downarrow 0$, follows from dominated
convergence. To this end we normal order the creation and annihilation
operators and obtain
$$
\psi_m(\eta) = \sum_{l=0}^m \psi_{m,l}(\eta) ,
$$
where $\psi_{m,l}(\eta)$ is an element of $\hh^{\otimes_s l}$.
Thus the term indexed by $l$ contains $l$ creation operators, which originate
from  positions in the set $X$, whereas
the other operators on the vertices are contracted.
Explicitly, we obtain using the pull through formula and algebraic identities as
in the proof of Proposition~\ref{thm:algebraicenergies},
\begin{align}\label{eq:groundstate}
   \psi&_{m,l}(\eta)(p_1,p_2,\dotsc,p_l) \nonumber \\
&{} =  c_l  \sum_{\substack{X \subset     N_m \\ |X|=l  }} \sum_{\pi : N_l  \to X }  \sum_{Y \subset N_m \setminus X } \!\!\sum_{\substack{ P\  \text{pairing of } Y : \\ \text{for each } p \in P \\
\text{there exists an }  x \in X  \text{ and}  \\ \text{a linked path from  p}  \text{ to }  \{ 0 , x\} \\
 \text{in }  P \cup \{\{0,x\} \} }}   \nonumber  \\
  &  \qquad\;\; \prod_{y \in Y} \left\{\int dk_y\right\} \delta_P(k) \tilde{F}_{\{0,\min Z\}}(|K_{\{0,\min Z\}}|_{P \cup P_X} , \eta ) \nonumber  \\
 &\qquad\quad\;\; \times \!\!\!   \prod_{j \in Z\setminus \max Z} \left\{ G^\sharp_{j,P \cup P_X}(k_j  )
  \tilde{F}_{\{j,j_Z'\} }(|K_{\{j,r_Z(j)\}}|_{P \cup P_X} , \eta )  \right\} \nonumber \\
  &\qquad \qquad \quad\;  \times G^\sharp_{\max Z,P \cup P_X}(k_{\max Z}) \widehat{F}_{\max Z}(r,\eta) \psi_0  |_{\{ k_{\pi(s)} = p_s : s \in N_l \}  }   ,
\end{align}
where $c_l$ is a combinatorial factor, we have set $Z := Y \cup X $ and $P_X := \{ \{ 0 , x \} : x \in X \}$, and
we have used the notations introduced in the definition of the contraction
\begin{align*}
  G^\sharp_{j,P}   & := \begin{cases}
    G_j^*     & , \exists p \in P ,  \  j = \min p  \\
    G_j & , \exists p \in P , \ j = \max p  ,
  \end{cases} \\
|K_e |_P & := \sum_{\substack{ p \in P \\ \max e \leq \max p \\ \min p
    \leq \min e }} | k_{\max p}| ,
\end{align*}
and we have set
$$
\tilde{F}_{\{i,j\}}(r,\eta) := \begin{cases}
	R(r,\eta), & \text{if } j-i = 1, \\
	R(r,\eta) \widehat{T}_{j-i}(r,\eta)   R(r,\eta), &\text{otherwise,}
\end{cases}
$$
$$
\widehat{F}_{\max Z}(r) := \begin{cases}
	\boldsymbol{1}_{\HH}, & \text{if } \max Z = m, \\
	R(r)\widehat{T}_{m-\max Z}(r), & \text{otherwise.}
\end{cases}
$$

 Now observe that the integrand on the right hand side of \eqref{eq:groundstate} exists for $\eta=0$, this
 follows from the pull-through formula and that $k_j=0$ is a set of measure zero.
 A singularity in a possible factor on the very right vanishes because of the
  projection onto the orthogonal complement of the unperturbed ground state.
   From the   estimate  in the   proof of Lemma \ref{lem:lemestimate} we see that
    $\psi_{m,l}(0)$ is square integrable.
   Furthermore, using   the continuity of $T_{m,n}(r,\eta)$ on $(0,\infty) \times [0,1]$ and
 again the   estimate  in the   proof of Lemma \ref{lem:lemestimate} we see from dominated
 convergence that $\psi_{m,l}(\eta)  \to \psi_{m,l}(0)$ for $\eta \downarrow 0$.
\end{proof}

\section{Proof of Theorem  \ref{thm:groundstatee}}
\label{sec:proofofmain}

In this section we give a proof of Theorem  \ref{thm:groundstatee}.
First we  show that the ground state and the ground state
energy are continuous functions of the coupling constant, that is we verify  Hypothesis~\ref{hyp:1}.
We recall  the notation   $\psi_0 = \varphi_{\rm at} \otimes \Omega$ and $E_0 = E_{\rm at}$.

\begin{proposition} \label{prop:contgse}
Let $H(\lambda)$ be given as in \eqref{eq:defofh} and assume that Hypothesis  \ref{hyp:0} is satisfied.
Then the following holds.
\begin{enumerate}[label=\rm (\alph*)]
\item \label{prop:contgse:a} If  \eqref{eq:assonG} holds,
then the  ground state energy $E(\lambda)$ satisfies $E(\lambda) \leq E_0 $ and
$$
  E(\lambda) -  E_0 =  O(|\lambda|^2)   , \quad (\lambda \to  0 ).
$$
\item \label{prop:contgse:b}If \eqref{eq:assonG2} holds,
then the operator $H(\lambda)$ has an eigenvector   $\psi(\lambda)$ with eigenvalue $E(\lambda)$ such that
$$
\|  \psi(\lambda) - \psi_0  \| = O(|\lambda|)  , \quad (\lambda \to  0 )
$$
and $\inn{\psi_0  , \psi(\lambda)} = 1$ for $\lambda$ in a neighborhood of zero.
\end{enumerate}
\end{proposition}
\begin{proof}
  \ref{prop:contgse:a}.
  First we   show  the upper
  bound
  \begin{equation*}
    E(\lambda)  \leq   \inn{ \psi_0 , H(\lambda) \psi_0 } = \inn{ \psi_0 ,
      (H_f + H_{\rm at} + \lambda \phi(G) ) \psi_0 } = E_{\rm at} =
    E_0 .
  \end{equation*}
  To show the lower bound we complete the square
  \begin{align*}
    H(\lambda) & = H_{\rm at} + H_f + \lambda \phi(G) \\
               & = H_{\rm at} + \int dk |k| \left[ a(k) + \lambda \frac{G(k)}{|k|} \right]^*\left[ a(k) + \lambda \frac{G(k)}{|k|} \right]
               \\ &\qquad \quad - |\lambda|^2   \int   \frac{G(k)^*G(k)}{|k|} dk  \\
               & \geq E_{\rm at}  - |\lambda|^2   \int   \frac{\|G(k)\|^2}{|k|} dk  .
  \end{align*}
  \ref{prop:contgse:b} This is a consequence of the following two claims. We write
  $\widehat{\psi}(\lambda) := \frac{\psi(\lambda)}{ \| \psi(\lambda) \|}$.

\vspace{0.5cm}
\noindent
  {\bf Claim 1:}  We have that $\| \bar{P}_\Omega \widehat{\psi}(\lambda) \|  = O(|\lambda|)$.

\vspace{0.5cm}
\noindent
  Calculating  a commutator we obtain
\begin{align*}
  H(\lambda)   a(k)  \psi(\lambda)   &=  ( [H(\lambda), a(k)] + a(k) H(\lambda) ) \psi(\lambda)    \\
                                     &=  ( -|k| a(k) -  \lambda G(k) + a(k) H(\lambda) ) \psi(\lambda)    .
\end{align*}
Solving for  $a(k)  \psi(\lambda)$ we find
\begin{align*}
  ( H(\lambda)  - E(\lambda)  + |k| ) a(k)  \psi(\lambda)
  &=   -  \lambda G(k)  \psi(\lambda) ,
\end{align*}
and by inversion we find for  $k \neq 0$ that
\begin{align*}
  a(k)  \psi(\lambda)
  &=   - \lambda  \frac{|k|}{ H(\lambda)   - E(\lambda) + |k| }  \frac{G(k)}{|k|}   \psi(\lambda) .
\end{align*}
Thus we obtain for the number operator $N$ the expectation
\begin{align*}
  \inn{ \psi(\lambda) , N \psi(\lambda) }
  &= \int dk \| a(k) \psi(\lambda)
    \|^2 \\
   &=  |\lambda|^2  \int
    dk  \left\| \frac{|k|}{
    H(\lambda)  - E(\lambda)
    + |k| }  \frac{G(k)}{|k|}
    \psi(\lambda)  \right\|^2
  \\
  &\leq  |\lambda|^2  \int dk   \frac{\| G(k)\|^2}{|k|^2} \| \psi(\lambda) \|^2  .
\end{align*}
Inserting this into the inequality
$$
\| \bar{P}_\Omega  \psi \|^2 \leq  \inn{ \psi , N \psi  }
$$
we find  that
\begin{equation*}
	\| \bar{P}_\Omega \widehat{\psi}(\lambda)  \| = O(\lambda) ,  \quad (\lambda \to 0 ) .
\end{equation*}
This shows Claim 1.

\vspace{0.5cm}
\noindent
 {\bf Claim 2:}  Let $\bar{P}_{\rm at} = 1 - P_{\rm at}$.
 Then we have $\| \bar{P}_{\rm at}  \widehat{\psi}(\lambda) \|  = O(|\lambda|)$.

\vspace{0.5cm}
\noindent
We apply $\bar{P}_{\rm at}$ to the eigenvalue equation and obtain
$$
\bar{P}_{\rm at} H(\lambda)  \bar{P}_{\rm at} \psi(\lambda) + \bar{P}_{\rm at} H(\lambda) {P}_{\rm at}
\psi(\lambda) = E(\lambda) \bar{P}_{\rm at} \psi(\lambda) .
$$
Solving for terms involving $\bar{P}_{\rm at} \psi(\lambda)$ we find
\begin{equation} \label{eq:solvforpbarpsi}
( \bar{P}_{\rm at} H(\lambda) \bar{P}_{\rm at} - E(\lambda) \bar{P}_{\rm at} ) \bar{P}_{\rm at}
\psi(\lambda) = - \bar{P}_{\rm at} H(\lambda)  {P}_{\rm at} \psi(\lambda) .
\end{equation}
Below we want to show that we can invert the operator on the left and, moreover,
we want to estimate the inverse. To this end we will use a Neumann expansion.
Let $\epsilon_1 := \inf \sigma  (H_{\rm at} |_{\ran \bar{P}_{\rm at}}  )$.
By \ref{prop:contgse:a}  we have in the sense of operators  on the range of $\bar{P}_{\rm at}$  that
\begin{equation*}
( H(0) - E(\lambda) ) \bar{P}_{\rm at}  \geq
( H(0) - E_0 )
  \bar{P}_{\rm at}  = (H_{\rm at}  + H_f - E_0 ) \bar{P}_{\rm at}   \geq   (\epsilon_1 - E_0    )   \bar{P}_{\rm at}      .
\end{equation*}
Thus  $( H(0) - E(\lambda) )
  \bar{P}_{\rm at}$ is invertible as an operator in $\ran \bar{P}_{\rm at}$.
We note the standard estimates
\begin{align*}
\| a(G) \psi \| & \leq \left( \int \frac{ \| G(k)\|^2}{|k|} dk
\right)^{1/2} \| H_f^{1/2} \psi \| \\
\| a^*(G) \psi \|^2 &  \leq   \int  \| G(k)\|^2 dk  \| \psi \|^2 +  \int \frac{ \| G(k)\|^2}{|k|} dk  \| H_f^{1/2} \psi \|^2 ,
\end{align*}
which imply that
$$
 \| (H_f + 1 )^{-1/2} \phi(G) \|   =   \|  \phi(G)(H_f + 1 )^{-1/2} \|  < \infty .
$$
By \ref{prop:contgse:a}  we find that
\begin{align}
  \| (  \bar{P}_{\rm at}  ( H&(0) - E(\lambda) )  \bar{P}_{\rm at} )^{-1}  \bar{P}_{\rm at}  \phi(G)  \| \nonumber \\
    &\leq   \| (  \bar{P}_{\rm at} (  H(0) - E(\lambda) )  \bar{P}_{\rm at} )^{-1} (H_f + 1 )^{1/2} \| \| (H_f + 1 )^{-1/2}  \phi(G)  \|  \nonumber \\
   & \leq \sup_{r \geq 0} \left| \frac{(r+1)^{1/2}}{r + \epsilon_1 - E_0 } \right|   \| (H_f + 1 )^{-1/2}  \phi(G)  \|     =:  C_G .
 \label{eq:lowboundcomp2}
\end{align}
By Neumanns Theorem  it follows from  \eqref{eq:lowboundcomp2}   that
$\bar{P}_{\rm at} ( H(\lambda) - E(\lambda) ) \bar{P}_{\rm at} $ is
invertible
on $\ran \bar{P}_{\rm at}$, if $|\lambda| <  C_G^{-1}$, and
\begin{align}
  ( & \bar{P}_{\rm at}  ( H(\lambda) - E(\lambda) )  \bar{P}_{\rm at}
  )^{-1} \nonumber \\
  &= \sum_{n=0}^\infty \left[ -  (  \bar{P}_{\rm at}  ( H(0) - E(\lambda) )  \bar{P}_{\rm at} )^{-1} \lambda \phi(G) \right]^n
  (  \bar{P}_{\rm at}  ( H(0) - E(\lambda) )  \bar{P}_{\rm at} )^{-1} .
 \label{eq:lowboundcomp3}
\end{align}
Inserting  \eqref{eq:lowboundcomp3} into \eqref{eq:solvforpbarpsi} and using again   \eqref{eq:lowboundcomp2}   we find
\begin{align*}
\| \bar{P}_{\rm at} \widehat{\psi}(\lambda) \| & = \| [  \bar{P}_{\rm at} (  H(\lambda) - E(\lambda) )  \bar{P}_{\rm at} ]^{-1}
   \bar{P}_{\rm at} H(\lambda)  {P}_{\rm at} \widehat{\psi}(\lambda) \|  \\
   & \leq  \frac{|\lambda| C_G}{1- |\lambda| C_G} \|  {P}_{\rm at} \widehat{\psi}(\lambda)\|  .
\end{align*}
This shows Claim 2.

\vspace{0.5cm}
\noindent
\ref{prop:contgse:b} now follows from Claims 1 and 2 by writing
\begin{align*}
  \widehat{\psi}(\lambda) - \psi_0 \inn{ \psi_0 ,   \widehat{\psi}(\lambda)  }
  &= \widehat{\psi}(\lambda) - P_\Omega \otimes  P_{\rm at}
  \widehat{\psi}(\lambda) \\
  &= \bar{P}_\Omega   \widehat{\psi}(\lambda)   + P_\Omega  \otimes \bar{P}_{\rm at}  \widehat{\psi}(\lambda) \to 0  ,
\end{align*}
where the first term on the right hand side tends to zero because of Claim~1 and
 the second term because of Claim 2.
 Now  $\psi(\lambda) = \widehat{\psi}(\lambda) \inn{ \psi_0 ,
   \widehat{\psi}(\lambda) }^{-1}$
 is well defined for $\lambda$ sufficiently close to zero and satisfies \ref{prop:contgse:b}.
\end{proof}

\begin{proof}[Proof of Theorem \ref{thm:groundstatee}]
First we show using  Theorems
  \ref{thm:mainenergy} and \ref{thm:groundstate}
that
  \begin{equation} \label{eq:inducstateend0}
H_0 \psi_{n+1}(0) + V \psi_n(0) = \sum_{k=0}^{n+1} E_k \psi_{n + 1  - k}(0)  .
  \end{equation}
From the convergence of $\psi_n(\eta)$ as $\eta \downarrow 0$ we obtain
from the definition of $E_n$ that
  \begin{equation} \label{eq:inducstateend1}
E_n = \inn{ V  \psi_0 ,\psi_n(0 ) } =  \lim_{\eta   \downarrow 0} \inn{ V  \psi_0 , \psi_n(\eta ) } .
  \end{equation}
From the definition of  $\psi_n(\eta)$  (compare  \eqref{eq:inducstate})
we see that
  \begin{equation} \label{eq:inducstateend}
  (H_0 - E_0 + \eta)  \psi_{n+1}(\eta) =
    \bar{P}_0 \left( \sum_{k=1}^{n+1} E_k \psi_{n+1-k}(\eta) - V \psi_{n}(\eta)
    \right)
  \end{equation}
  We claim that the limit $\eta \downarrow 0$ yields
    \begin{equation} \label{eq:inducstateend2}
  (H_0 - E_0)  \psi_{n+1}(0) =
    \bar{P}_0 \left( \sum_{k=1}^{n+1} E_k \psi_{n+1-k}(0) - V \psi_{n}(0)
    \right)
  \end{equation}
  This clearly holds for $n=0$. Suppose that it holds for all $m \leq n-1$. Then for $n$ the
  right hand side of  \eqref{eq:inducstateend}  converges to the right hand side of \eqref{eq:inducstateend2}. Since $H_0$ is a closed operator it follows that the left hand
  side of \eqref{eq:inducstateend} converges to the left hand side of
   \eqref{eq:inducstateend2}.  Now  \eqref{eq:inducstateend2}
   and \eqref{eq:inducstateend1}  imply  \eqref{eq:inducstateend0}.
  By Proposition
  \ref{prop:contgse} and  \eqref{eq:inducstateend0} the assumptions
  of Lemma   \ref{lem:abstractmain}
 are  satisfied.  Hence
   Theorem  \ref{thm:groundstatee} now follows from Lemma \ref{lem:abstractmain}.
\end{proof}

\bibliography{asymptotic}
\bibliographystyle{amsplain}

\end{document}